\documentclass[aps,prd,reprint,nofootinbib,longbibliography]{revtex4-2}
\pdfoutput=1

\usepackage{amsmath,amssymb,mathtools,bm,amsthm,mathrsfs}
\usepackage{graphicx}
\usepackage[hidelinks]{hyperref}

\newtheorem{theorem}{Theorem}
\newtheorem{lemma}{Lemma}
\newtheorem{corollary}{Corollary}

\newcommand{\SPD}{\mathrm{SPD}}
\newcommand{\Tr}{\operatorname{tr}}
\newcommand{\Id}{\mathbb{I}}
\newcommand{\dd}{\mathrm{d}}
\newcommand{\cQ}{\mathcal{Q}}
\newcommand{\cH}{\mathcal{H}}
\newcommand{\DT}{\mathcal{D}_{T}}
\newcommand{\Dini}{D^{+}}

\begin{document}

\title{Renormalization-group locking of finite anisotropic propagation cones in Yukawa networks}

\author{Shuai Zeng}
\email{zengshuai@cqupt.edu.cn}
\affiliation{Chongqing University of Posts and Telecommunications, Chongqing 400065, China}

\begin{abstract}
Different field species can begin with finite anisotropic propagation
geometries whose principal axes need not coincide.  We ask whether local
Yukawa interactions on a finite graph can nevertheless generate one common
infrared spatial cone, even when the interactions that drive the locking
become marginally irrelevant.  Retaining the full finite mismatch between positive-definite spatial kinetic
matrices, we derive the one-loop nonlinear matrix flow and show that its
Thompson diameter is globally nonexpansive.  On a connected graph, persistent normalized edge activity
upgrades this to a uniform finite-window contraction, so all interacting
sectors asymptotically share one relative propagation geometry.  An explicit
two-channel Pauli--Dirac Yukawa--quartic completion realizes the persistence
conditions on an open weak-coupling set of finite-mismatch initial data.  In four dimensions the Yukawa
couplings vanish in the infrared, yet their accumulated interaction time
diverges, leaving a Gaussian endpoint with a common spatial cone.  For a
controlled two-plateau extremal class, graph distance also fixes the first
nonzero short-time order of the diameter decrease.  Thus common-cone locking can emerge from renormalization-group dynamics
rather than being imposed as a microscopic common structure.

\end{abstract}

\maketitle

\section{Introduction}
\label{sec:introduction}

A recurring question in quantum field theory is why sectors with distinct
microscopic propagation laws can share a common low-energy characteristic
structure.  In its simplest form this is the problem of a universal limiting
speed: Lorentz-violating theories may begin with different fermion, boson, or
gauge velocities, while interactions drive those velocities toward a common
infrared value \cite{ChadhaNielsen1983,AnberDonoghue2011}.  Related mechanisms
appear in strongly coupled settings and near interacting quantum critical
points \cite{BednikPujolasSibiryakov2013,RoyJuricicHerbut2016}.  The underlying question extends beyond equality of a few velocity parameters.
A generic anisotropic sector carries an entire characteristic quadratic form,
and a common low-energy cone requires the relative propagation geometry of all
interacting species to become compatible.

Renormalization-group studies have progressively enlarged the kinematic data
included in this problem.  Species-dependent limiting velocities and emergent
Lorentz symmetry have been analyzed in Yukawa and related interacting theories
\cite{AnberDonoghue2011,RoyJuricicHerbut2016,KharukSibiryakov2016}.
Anisotropic Dirac cones, multiple nodes, and anisotropic fixed points have
been studied in fermionic and gauge-coupled systems
\cite{VafekTesanovicFranz2002,IsobeNagaosa2013,WangLiu2014,
PozoFerreirosVozmediano2018}.  Complementary developments include tilted,
birefringent, and non-Hermitian propagation
\cite{ReiserJuricic2024,MurshedRoy2025,PinoAlarconJuricic2025}, explicit
one-loop renormalization of Lorentz-violating scalar and Yukawa kinetic
operators \cite{FerreroAltschul2011}, and general anisotropic quadratic
structures \cite{AgeevAgeeva2026}.  These developments motivate a
finite-mismatch problem in which the complete relative propagation matrices
and their eigenframes evolve dynamically across interacting sectors.  The
question addressed here is whether this RG matching mechanism extends to a
finite family of full spatial quadratic forms at finite mismatch, including
noncoincident eigenframes, and whether local edgewise attraction enforces
global common-cone locking on a sparse interaction network.

We formulate that problem in terms of spatial propagation matrices
\begin{equation}
G_i\in\SPD(3),
\label{eq:intro-Gi}
\end{equation}
where \(i\) labels a field species and \(\SPD(3)\) denotes the cone of real
symmetric positive-definite \(3\times3\) matrices.  We write \(A_a\) for a
fermion propagation matrix and \(B_r\) for a scalar propagation matrix, with
\(a\) and \(r\) labeling fermionic and scalar species, respectively.
A full SPD matrix encodes both principal propagation rates and their
eigenframes.  For a family of three or more sectors, a single spatial congruence may fail to
diagonalize all of them simultaneously.  The RG
problem is then genuinely matrix valued: the infrared dynamics must align
relative eigenvalues and relative eigenframes across the full family.

Interaction topology introduces a second ingredient.  We consider a
finite bipartite Yukawa graph whose fermion and scalar nodes carry the
matrices \(A_a\) and \(B_r\).  Global network contraction requires edgewise
attraction to propagate through shared nodes.  On a sparse graph an extremal
node can initially be shielded by neighboring cones that are equal along the
active extremal direction, so the first derivative of a global mismatch can
vanish even though another edge farther along the graph is already
contracting.  Global locking therefore requires control of the finite-time propagation of
strictness through the interaction network.

A third ingredient is field-theoretic persistence.  The cone-flow rates are
proportional to running Yukawa couplings and wave-function factors.  In a
marginal four-dimensional theory these couplings can decrease toward the
infrared.  The QFT application then requires the coupling RG to supply an infinite
accumulated interaction time while maintaining persistent relative edge
activity.  These ingredients lead to the central question of
this work: can local Yukawa interactions erase finite, noncommuting,
species-dependent propagation geometry across an entire sparse network even
when the interactions responsible for the locking themselves vanish in the
infrared?

The present paper answers this question within a multichannel
Dirac--scalar Yukawa theory.  First, we retain the full finite relative SPD
mismatch in the one-loop self-energies.  For one fermion--scalar edge with fermion matrix \(A\) and scalar matrix
\(B\), the relative matrix
\begin{equation}
R=A^{-1/2}BA^{-1/2}
\end{equation}
enters the nonlinear kernel
\begin{equation}
\mathcal H(R)
=
\int_0^1\dd y\,
\frac{(1-y)
\left[
\Id-\left((1-y)\Id+yR\right)^{-1}
\right]}
{\sqrt{\det[(1-y)\Id+yR]}},
\label{eq:intro-H}
\end{equation}
where \(\Id\) is the \(3\times3\) identity matrix and \(y\in[0,1]\) is
the Feynman parameter.  The congruence-covariant perspective of this kernel
generates the fermion propagation flow.  Second, special identities of this
Yukawa kernel imply global nonexpansion of the Thompson diameter
\(\mathcal D_T\), defined in section~\ref{subsec:thompson}.  Writing
\(D^+\) for the upper-right Dini derivative, this result is
\begin{equation}
D^+\mathcal D_T\le0.
\label{eq:intro-nonexpansion}
\end{equation}
The inequality holds for arbitrary finite positive-definite mismatch in the
controlled one-loop domain.  Third, on a fixed connected graph with persistent normalized edge
activity, equality rigidity and positive propagation of endpoint strictness
upgrade nonexpansion to a finite-window contraction.  Here \(s\) denotes the
accumulated interaction time, \(T>0\) a fixed window length, and
\(\rho_T\) the corresponding contraction factor:
\begin{equation}
\mathcal D_T(s+T)
\le
\rho_T\mathcal D_T(s),
\qquad
0<\rho_T<1,
\label{eq:intro-window}
\end{equation}
and hence
\begin{equation}
\mathcal D_T(s)\longrightarrow0.
\label{eq:intro-locking}
\end{equation}
This is the global common-cone locking result.  We further show in an
explicit two-channel Pauli--Dirac Yukawa--quartic completion that the
coupling RG itself generates a compact weak-coupling persistence basin on an
open set of genuinely finite-mismatch initial data.  The resulting mechanism goes beyond pairwise velocity matching: local Yukawa
interactions synchronize both relative propagation rates and relative
eigenframes across the connected network.  Figure~\ref{fig:intro-schematic} summarizes this
mechanism.  Connectivity determines the asymptotic locking of the network,
while graph distance controls when strictness first becomes visible at
shielded extremal sectors in the controlled two-plateau class.

\begin{figure*}[t]
\centering
\includegraphics[width=0.96\textwidth]{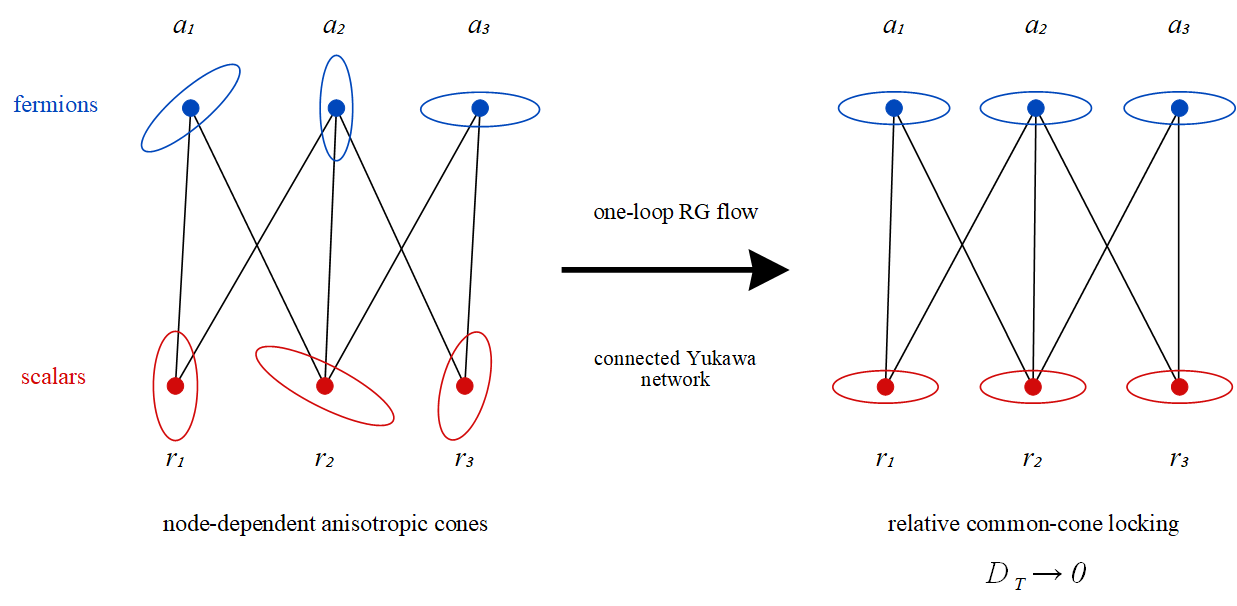}
\caption{Conceptual overview of common-cone locking in a connected Yukawa
network.  Left: fermion and scalar nodes carry node-dependent anisotropic
spatial propagation cones.  Middle: the one-loop RG acts on the connected
bipartite interaction graph.  Right: persistent connected interactions drive
relative common-cone locking, quantified by \(\mathcal D_T\to0\).
Connectivity underlies the global-locking theorem, while graph distance
controls the onset order in the topology-sensitive short-time law.}
\label{fig:intro-schematic}
\end{figure*}

The interaction graph also leaves a sharper short-time imprint.  For a
controlled class of two-plateau extremal initial states, the first nonzero
decrease of the Thompson diameter occurs at an accumulated-time order fixed
by graph distance.  For a common simple extremal direction \(x\), let
\(L_x\) and \(U_x\) denote the lower and upper initial plateau node sets,
let \(\mathcal A_x\) be the set of initially active ordered endpoint pairs,
and let \(d(v,S)\) denote the shortest graph distance from a node \(v\) to
a node set \(S\).  The relevant shielding depth is
\begin{equation}
q
=
\max_{(i,j)\in\mathcal A_x}
\min\{d(i,U_x),d(j,L_x)\},
\label{eq:intro-q}
\end{equation}
and
\begin{equation}
\mathcal D_T(0)-\mathcal D_T(s)
=
C_qs^q+O(s^{q+1}),
\qquad
C_q>0.
\label{eq:intro-onset}
\end{equation}
Here \(C_q\) is the positive leading coefficient and \(O(s^{q+1})\) is
Landau notation for the higher-order remainder as \(s\to0^+\).
This result refines the asymptotic locking theorem by separating two roles of
the interaction graph: connectivity determines the asymptotic fate of the
network, whereas graph distance determines the earliest perturbative order at
which strict contraction reaches shielded extremal sectors.

The analysis is carried out at one loop and weak coupling for a fixed Yukawa
graph, positive-definite spatial kinetic matrices, and vanishing time--space
tilt.  In this perturbative setting, the flow resolves the complete relative
SPD matrix at finite cone mismatch.
Section~\ref{sec:yukawa-flow} derives
the finite-matrix RG flow.  Section~\ref{sec:global-geometry} develops the
relative-cone geometry and proves global nonexpansion.
Section~\ref{sec:finite-window-locking} proves finite-window locking and
constructs the explicit field-theoretic persistence basin.
Section~\ref{sec:topology-onset} derives the topology-sensitive onset law,
and section~\ref{sec:discussion} discusses the physical interpretation,
domain, and extensions.  Technical derivations and proofs are collected in
appendices~\ref{app:internal-algebra}--\ref{app:persistence-topology}.

\section{Multichannel Yukawa theory and finite-mismatch RG}
\label{sec:yukawa-flow}

We derive the propagation-geometry flow directly from a multichannel
Dirac--scalar Yukawa theory.  The construction identifies a finite class of
field theories with a closed one-loop two-point sector, while retaining the
full relative spatial kinetic matrices at finite mismatch.  The resulting
matrix-valued RG equations provide the starting point for the global analysis
in section~\ref{sec:global-geometry}.

\subsection{Fields, propagation matrices, physical cones, and graph}
\label{subsec:fields-graph}

Let
\[
\mathsf G=(V_F\sqcup V_B,E)
\]
be a finite bipartite interaction graph, with fermionic node set \(V_F\),
bosonic node set \(V_B\), and edge set \(E\).  We write \(r\sim a\) when
\((a,r)\in E\).  A fermionic node \(a\in V_F\) carries a four-component
Dirac field \(\psi_a\) with a finite-dimensional internal flavor space
\(\mathcal V_a\), while a bosonic node \(r\in V_B\) carries a real scalar
field \(\phi_r\).  We denote the scalar multiplet by
\(\boldsymbol\phi=(\phi_r)_{r\in V_B}\).  The matrices
\(\gamma_\mu\) are Euclidean Dirac matrices, \(\bar\psi_a\) is the Dirac
adjoint, \(\partial_0=\partial/\partial\tau\) is the Euclidean-time
derivative, and repeated spatial indices \(i,j=1,2,3\) are summed.  At a
massless critical point in three spatial dimensions, we consider the Euclidean
action
\begin{align}
S
={}&
\int \dd\tau\,\dd^3x\,
\Bigg[
\sum_{a\in V_F}
\bar\psi_a
\left(
\gamma_0\partial_0
+
\gamma_i e^{(a)}_{ij}\partial_j
\right)\psi_a
\nonumber\\
&\qquad
+
\frac12\sum_{r\in V_B}
\left(
(\partial_0\phi_r)^2
+
\partial_i\phi_r\,B_r^{ij}\partial_j\phi_r
\right)
+
V(\boldsymbol{\phi})
\nonumber\\
&\qquad
+
\sum_{(a,r)\in E}
g_{ar}\,
\phi_r\,
\bar\psi_a T_{ar}\psi_a
\Bigg].
\label{eq:jhep-action}
\end{align}
Here \(e_a=[e^{(a)}_{ij}]\) is a real invertible spatial kinetic frame,
\(B_r\in\SPD(3)\), \(g_{ar}\) is the real Yukawa coupling on edge
\((a,r)\), and the matrices \(T_{ar}\) act only on \(\mathcal V_a\) and
commute with the spacetime Dirac algebra.  The scalar potential \(V(\boldsymbol\phi)\) contains the interactions that
complete the critical theory.  The one-loop propagation-matrix
flow derived below is determined by the Yukawa sector.  For the explicit
two-channel persistence realization, the complete quartic sector is specified
in section~\ref{subsec:qft-persistence} and
appendix~\ref{appsub:quartic-completion}.

The spatial quadratic forms carried by the two sectors are
\begin{equation}
A_a=e_a^{T}e_a\in\SPD(3),
\qquad
B_r\in\SPD(3).
\label{eq:jhep-cone-matrices}
\end{equation}
These are also the matrices that control the physical characteristic
quadratic forms after analytic continuation.  With the temporal kinetic
coefficients normalized to unity and time--space tilt set to zero, let
\(\omega\) denote frequency and
\(\boldsymbol k\in\mathbb R^3\) the spatial momentum.  The massless scalar
dispersion relation is
\begin{equation}
\omega^2=\boldsymbol{k}^{T}B_r\boldsymbol{k},
\label{eq:jhep-boson-dispersion}
\end{equation}
whereas the determinant of the fermionic kinetic operator gives
\begin{equation}
\omega^2=\boldsymbol{k}^{T}A_a\boldsymbol{k}.
\label{eq:jhep-fermion-dispersion}
\end{equation}
Thus \(A_a\) and \(B_r\) determine spatial sections of the characteristic
cones.  Throughout the paper, common-cone locking means synchronization of
these relative characteristic quadratic forms.

The graph \(\mathsf G\) records which sectors exchange Yukawa fluctuations:
\((a,r)\in E\) precisely when the coupling \(g_{ar}\) is present.  The graph is
kept fixed in the one-loop problem studied here; the edge strengths may run
with the RG scale.

\subsection{Channel closure and finite-degree realization}
\label{subsec:channel-closure}

For every edge incident on a fermionic node \(a\), choose a Hermitian
internal matrix satisfying
\begin{equation}
T_{ar}^{\dagger}=T_{ar},
\qquad
T_{ar}^{2}=\Id_{\mathcal V_a},
\label{eq:jhep-involution}
\end{equation}
and, for distinct scalar neighbors \(r,s\sim a\),
\begin{equation}
\Tr_{\mathcal V_a}(T_{ar}T_{as})=0.
\label{eq:jhep-HS-orthogonality}
\end{equation}
Here \(\Id_{\mathcal V_a}\) denotes the identity on the internal flavor
space and \(\Tr_{\mathcal V_a}\) denotes its matrix trace.  These conditions ensure that the propagation sector closes edge by edge at
one loop.  For the fermion self-energy associated with edge \((a,r)\), the
internal factor is \(T_{ar}^2=\Id_{\mathcal V_a}\), so contributions from
different incident edges add within the original fermion kinetic sector.  A mixed scalar two-point
function \(\phi_r\)--\(\phi_s\) generated through the same fermion is
proportional to
\(\Tr(T_{ar}T_{as})\) and therefore vanishes for \(r\neq s\).  Scalar quartic tadpoles are momentum independent at this order and contribute
to the local potential/mass sector.  Consequently, the logarithmically
divergent momentum-dependent two-point sector closes on the prescribed
bipartite graph.

The required internal algebra exists at every finite graph degree.  If
\(d_a=\deg(a)\), choose \(n_a\) such that
\begin{equation}
4^{n_a}-1\ge d_a
\label{eq:jhep-pauli-count}
\end{equation}
and assign distinct nonidentity \(n_a\)-qubit Pauli strings to the incident
edges.  Such strings are Hermitian involutions and are mutually
Hilbert--Schmidt orthogonal.  Matrix-valued order-parameter channels of this
type are familiar in multicomponent Gross--Neveu--Yukawa constructions
\cite{HerbutScherer2022}.

For the Pauli-string subclass used below, the Yukawa vertex directions also
remain closed at one loop.  Any two Pauli strings either commute or
anticommute, and therefore
\begin{equation}
T_{as}T_{ar}T_{as}
=
\eta^{(a)}_{sr}T_{ar},
\qquad
\eta^{(a)}_{sr}\in\{+1,-1\}.
\label{eq:jhep-pauli-conjugation}
\end{equation}
The cubic flavor structures entering the general one-loop scalar--fermion
Yukawa beta function are consequently proportional to the original external
channel:
\begin{equation}
\begin{aligned}
T_{as}^{2}T_{ar}&=T_{ar},
& T_{ar}T_{as}^{2}&=T_{ar},\\
T_{as}T_{ar}T_{as}&=\eta^{(a)}_{sr}T_{ar}.
\end{aligned}
\label{eq:jhep-vertex-closure}
\end{equation}
Together with the absence of off-diagonal scalar wave-function mixing from
eq.~\eqref{eq:jhep-HS-orthogonality}, this keeps the prescribed Pauli
Yukawa-channel subspace invariant at one loop.  The complete Dirac-to-Weyl
bookkeeping and its relation to the general tensorial one-loop Yukawa beta
function \cite{PannellStergiou2023} are collected in
appendix~\ref{app:internal-algebra}.  The anisotropic kinetic matrices modify
the scalar loop coefficients but not the internal Pauli algebra in
eqs.~\eqref{eq:jhep-pauli-conjugation}--\eqref{eq:jhep-vertex-closure}.

\subsection{Fermion self-energy at finite relative mismatch}
\label{subsec:fermion-self-energy}

Consider one edge and denote its fermion and scalar spatial quadratic forms
by \(A\) and \(B\).  A common spatial congruence, followed by a spin-frame
rotation, places the fermion kinetic term in the identity frame.  The only
relative spatial matrix is then
\begin{equation}
R=A^{-1/2}BA^{-1/2}\in\SPD(3).
\label{eq:jhep-relative-R}
\end{equation}
Up to overall Euclidean factors that leave the logarithmic counterterm
unchanged, the propagators in this frame may be written as
\begin{equation}
S(k)
=
\frac{\gamma_0k_0+\boldsymbol{\gamma}\!\cdot\!\boldsymbol{k}}
{k_0^2+\boldsymbol{k}^2},
\qquad
D_R(q)
=
\frac{1}{q_0^2+\boldsymbol{q}^{T}R\boldsymbol{q}}.
\label{eq:jhep-propagators}
\end{equation}
For external momentum \(p\), the one-loop fermion self-energy contains
\begin{align}
\Sigma(p)
={}&
g^2
\int
\frac{\dd k_0\,\dd^3k}{(2\pi)^4}
\left(\gamma_0 k_0+\boldsymbol{\gamma}\!\cdot\!\boldsymbol{k}\right)
\nonumber\\
&\times
(k_0^2+\boldsymbol{k}^2)^{-1}
\nonumber\\
&\times
\left[
(p_0-k_0)^2+
(\boldsymbol{p}-\boldsymbol{k})^{T}
R
(\boldsymbol{p}-\boldsymbol{k})
\right]^{-1}.
\label{eq:jhep-selfenergy}
\end{align}
All dependence on the internal channel has already reduced to \(T^2=\Id\).

Using a Feynman parameter \(y\in[0,1]\), define
\begin{equation}
M_y(R)=(1-y)\Id+yR.
\label{eq:jhep-My}
\end{equation}
The shifts
\begin{equation}
q_0=k_0-y p_0,
\qquad
\boldsymbol{q}
=
\boldsymbol{k}
-
yM_y(R)^{-1}R\boldsymbol{p}
\label{eq:jhep-loop-shift}
\end{equation}
remove terms linear in the loop momentum from the combined denominator.
Odd terms integrate to zero, while the external-momentum part of the numerator
becomes
\begin{equation}
y\gamma_0p_0
+
y\gamma_i
\bigl[M_y(R)^{-1}R\bigr]_{ij}p_j.
\label{eq:jhep-external-numerator}
\end{equation}
The logarithmically divergent scalar integral has the form
\begin{align}
I(M,\Delta)
&=
\int
\frac{\dd q_0\,\dd^3q}{(2\pi)^4}
\frac{1}
{\left[
q_0^2+\boldsymbol{q}^{T}M\boldsymbol{q}+\Delta
\right]^2}
\nonumber\\
&=
\frac{\kappa_4}{\sqrt{\det M}}
\log\frac{\Lambda_{\rm UV}}{\mu}
+
O(1).
\label{eq:jhep-log-integral}
\end{align}
For the spherical-cutoff convention
\(\log(\Lambda_{\rm UV}/\mu)\) used here,
\(\kappa_4=1/(8\pi^2)\).  The matrix dependence of the logarithmic
coefficient follows immediately from
\(\boldsymbol{Q}=M^{1/2}\boldsymbol{q}\); the same determinant dependence
is also obtained by proper-time regularization in
appendix~\ref{app:selfenergies}.

It follows that the logarithmic self-energy can be separated into temporal
and spatial pieces,
\begin{equation}
\Sigma_{\log}(p)
=
\kappa_4g^2
\log\frac{\Lambda_{\rm UV}}{\mu}
\left[
f_0(R)\gamma_0p_0
+
\gamma_iF_{ij}(R)p_j
\right],
\label{eq:jhep-sigma-log}
\end{equation}
where
\begin{align}
f_0(R)
&=
\int_0^1
\frac{y\,\dd y}{\sqrt{\det M_y(R)}},
\label{eq:jhep-f0}\\
F(R)
&=
\int_0^1
\frac{
yM_y(R)^{-1}R\,\dd y
}{
\sqrt{\det M_y(R)}
}.
\label{eq:jhep-F}
\end{align}

\subsection{The finite-mismatch matrix kernel}
\label{subsec:finite-kernel}

After restoring the temporal fermion kinetic coefficient to unity, the
relative spatial renormalization is governed by the difference
\(F(R)-f_0(R)\Id\).  We define
\begin{align}
\cH(R)
&=
F(R)-f_0(R)\Id
\nonumber\\
&=
\int_0^1\dd y\,
\frac{
(1-y)
\left[
\Id-M_y(R)^{-1}
\right]
}{
\sqrt{\det M_y(R)}
}.
\label{eq:jhep-H-kernel}
\end{align}
The second line follows from
\begin{equation}
yM_y(R)^{-1}R
=
\Id-(1-y)M_y(R)^{-1}.
\label{eq:jhep-My-identity}
\end{equation}
Equation~\eqref{eq:jhep-H-kernel} retains the complete positive-definite
relative matrix \(R\) at one-loop order throughout the perturbatively
controlled finite-mismatch domain.

Since every \(M_y(R)\) is a polynomial function of \(R\),
\(\cH(R)\) is a spectral matrix function of \(R\).  In particular,
\begin{equation}
[\cH(R),R]=0.
\label{eq:jhep-H-commutes}
\end{equation}
This property makes the one-edge relative dynamics diagonal in the
instantaneous eigenbasis of \(R\), while the many-node problem remains
genuinely matrix valued because different relative matrices need not share an
eigenframe.

\subsection{Scalar self-energy and congruence-covariant network flow}
\label{subsec:network-flow}

For a fermion cone \(A=e^{T}e\), the spatial change of variables
\(\boldsymbol{K}=e\boldsymbol{k}\) reduces the one-loop scalar bubble to an
isotropic fermion integral.  Its logarithmically divergent
momentum-dependent part has the quadratic form
\begin{equation}
\Pi_{\log}(p)
\propto
\frac{g^2\,\dim(\mathcal V_a)}{|\det e|}
\left(
p_0^2+\boldsymbol{p}^{T}A\boldsymbol{p}
\right)
\log\frac{\Lambda_{\rm UV}}{\mu},
\label{eq:jhep-scalar-selfenergy-structure}
\end{equation}
with a positive convention-dependent coefficient proportional to the
internal flavor multiplicity on the edge.  Restoring the scalar temporal
kinetic coefficient to unity therefore gives the edge flow
\begin{equation}
\dot B=\alpha(A-B),
\qquad
\alpha>0.
\label{eq:jhep-B-edge}
\end{equation}
The full trace and normalization are recorded in
appendix~\ref{app:selfenergies}.  The global analysis uses the positivity of
\(\alpha\) and its regular dependence on the remaining perturbative RG
variables.

To restore an arbitrary fermion frame, define the matrix perspective
\begin{equation}
\cQ(A,B)
=
A^{1/2}
\cH\!\left(
A^{-1/2}BA^{-1/2}
\right)
A^{1/2}.
\label{eq:jhep-Q-kernel}
\end{equation}
It is equivariant under a common spatial congruence:
\begin{equation}
\cQ(SAS^{T},SBS^{T})
=
S\cQ(A,B)S^{T}
\qquad
(S\in GL(3,\mathbb R)).
\label{eq:jhep-Q-covariance}
\end{equation}
Indeed,
\begin{equation}
U=(SAS^{T})^{-1/2}SA^{1/2}
\label{eq:jhep-U}
\end{equation}
is orthogonal, so the two relative matrices are orthogonally similar.
Orthogonal equivariance of the spectral function \(\cH\), together with
\((SAS^{T})^{1/2}U=SA^{1/2}\), gives
eq.~\eqref{eq:jhep-Q-covariance}.

Writing the induced fermion edge flow in terms of the quadratic form
\(A=e^{T}e\) gives
\begin{equation}
\dot A=2\beta\,\cQ(A,B),
\qquad
\beta>0.
\label{eq:jhep-A-edge}
\end{equation}
Summing the edge contributions over the bipartite graph yields
\begin{subequations}
\label{eq:jhep-network-flow}
\begin{align}
\dot A_a
&=
2\sum_{r\sim a}
\beta_{ar}\,
\cQ(A_a,B_r),
\label{eq:jhep-network-A}\\
\dot B_r
&=
\sum_{a\sim r}
\alpha_{ar}\,
(A_a-B_r).
\label{eq:jhep-network-B}
\end{align}
\end{subequations}
A dot denotes differentiation with respect to an infrared-oriented RG
parameter \(\ell\).  The positive rates \(\alpha_{ar}\) and \(\beta_{ar}\)
absorb \(g_{ar}^2\), internal multiplicities, determinant factors and the
positive loop normalization; they may themselves run with the remaining
couplings.  Equations~\eqref{eq:jhep-network-flow} are the
QFT-derived nonlinear matrix flow analyzed in the rest of the paper.

\subsection{Isotropic limit and transition to the global problem}
\label{subsec:isotropic-limit}

The finite-matrix kernel contains the known scalar-velocity result as a
special case.  For
\begin{equation}
R=\nu^2\Id
\end{equation}
in three spatial dimensions,
\begin{equation}
\cH(\nu^2\Id)
=
h(\nu^2)\Id,
\qquad
h(\nu^2)
=
\frac{4(\nu-1)}
{3\nu(\nu+1)^2},
\label{eq:jhep-isotropic-H}
\end{equation}
while the temporal coefficient is
\begin{equation}
f_0(\nu^2\Id)
=
\frac{2}{\nu(\nu+1)^2}.
\label{eq:jhep-isotropic-f0}
\end{equation}
The comparison with the scalar-velocity literature can be made explicitly.
Write
\begin{equation}
A=v_F^2\Id,
\qquad
B=v_B^2\Id,
\qquad
\nu=\frac{v_B}{v_F}.
\label{eq:jhep-rjh-isotropic-map}
\end{equation}
The fermion edge rate contains the Jacobian factor
\(|\det e|^{-1}=v_F^{-3}\).  Consequently, after stripping the common
positive loop normalization, eqs.~\eqref{eq:jhep-network-A} and
\eqref{eq:jhep-isotropic-H} give the Yukawa contribution
\begin{equation}
\dot v_F
\ \propto\
\frac{4g^2}{3}
\frac{v_B-v_F}{v_B(v_F+v_B)^2},
\label{eq:jhep-rjh-vF-match}
\end{equation}
which reproduces the full dependence on positive \(v_F\) and \(v_B\) of the
Yukawa term in eq.~(3.6) of Ref.~\cite{RoyJuricicHerbut2016} in its
infrared RG convention.  Likewise, the scalar equation gives
\begin{equation}
\dot v_B
\ \propto\
\frac{g^2}{2}
\frac{v_F^2-v_B^2}{v_F^3v_B},
\label{eq:jhep-rjh-vB-match}
\end{equation}
with the fermion-flavor multiplicity absorbed into \(\alpha\); this is the
velocity dependence of the Yukawa term in eq.~(3.7) of that reference.
Thus the finite-mismatch kernel reproduces the known one-loop flow for an
arbitrary positive velocity ratio, extending the comparison beyond the
linearized neighborhood of \(v_F=v_B\).  More general Lorentz-violating scalar--Yukawa theories are
also known to be one-loop renormalizable
\cite{FerreroAltschul2011}; the present construction keeps the complete
finite relative SPD matrix in the velocity-sector kernel.

The determinant factor in
eq.~\eqref{eq:jhep-log-integral} follows equivalently from a direct spatial
linear transformation and from proper-time regularization, while
eq.~\eqref{eq:jhep-Q-covariance} gives the corresponding
congruence-covariance identity.  The detailed derivation is collected in
appendix~\ref{app:selfenergies}.

The field theory thus supplies the full one-loop vector field on a network of
finite propagation matrices.  The remaining question is global:
whether this nonlinear flow admits a coordinate-independent measure of
relative cone mismatch that is controlled for arbitrary finite
positive-definite initial data.  We turn to that geometric problem in the
next section.

\section{Relative-cone geometry and global nonexpansion}
\label{sec:global-geometry}

Section~\ref{sec:yukawa-flow} reduced the field theory to a nonlinear flow on
a finite family of positive-definite propagation matrices.  We now ask a
coordinate-independent global question: can the relative spread of these
matrices grow under the Yukawa RG?  We first separate common-coordinate
freedom from genuine multi-cone incompatibility, then analyze one edge, and
finally prove that the Thompson diameter of the full network is
nonincreasing for arbitrary finite positive-definite mismatch.

\subsection{Common congruence and simultaneous diagonalization}
\label{subsec:relative-geometry}

A common change of spatial coordinates acts on every propagation matrix by
the same congruence,
\begin{equation}
G_i\longmapsto S G_i S^{T},
\qquad
S\in GL(3,\mathbb R),
\label{eq:common-congruence}
\end{equation}
where \(G_i\) denotes either a fermion matrix \(A_a\) or a scalar matrix
\(B_r\), and \(GL(3,\mathbb R)\) is the group of invertible real
\(3\times3\) matrices.  Relative propagation geometry is therefore naturally described by
quantities invariant under eq.~\eqref{eq:common-congruence}.  A standard
simultaneous-diagonalization characterization gives a convenient intrinsic
criterion for multi-cone compatibility.  Choose one reference cone \(G_0\)
and whiten the family,
\begin{equation}
H_i=G_0^{-1/2}G_iG_0^{-1/2}.
\label{eq:whitened-family}
\end{equation}
Then the full family \(\{G_0,G_1,\ldots,G_m\}\) can be diagonalized by one
common congruence if and only if
\begin{equation}
[H_i,H_j]=0
\qquad
\text{for all }i,j.
\label{eq:simultaneous-congruence-criterion}
\end{equation}
Indeed, if the \(H_i\) commute, their symmetry implies simultaneous
orthogonal diagonalizability; composing that orthogonal transformation with
\(G_0^{-1/2}\) gives the required common congruence.  Conversely, any common
congruence that diagonalizes the original family can first be rescaled so
that \(G_0\) is sent to the identity, leaving an orthogonal simultaneous
diagonalization of the whitened family and hence
eq.~\eqref{eq:simultaneous-congruence-criterion}.

This observation also identifies the smallest genuinely matrix-valued
network configuration.  Any two SPD matrices can be simultaneously
diagonalized by congruence, whereas the three-node Yukawa path
\begin{equation}
B_1\;-\;F\;-\;B_2
\label{eq:minimal-path}
\end{equation}
can already fail to admit a common diagonal frame.  Taking the fermion cone
\(A_F\) as reference, define
\begin{equation}
H_1=A_F^{-1/2}B_1A_F^{-1/2},
\qquad
H_2=A_F^{-1/2}B_2A_F^{-1/2}.
\label{eq:minimal-whitened}
\end{equation}
Whenever \([H_1,H_2]\neq0\), the three propagation geometries cannot be
reduced by one spatial congruence to independent scalar-velocity channels.
This is the finite-matrix setting addressed below.

\subsection{Single-edge attraction at finite mismatch}
\label{subsec:single-edge}

Before treating a network, consider a single Yukawa edge with fermion cone
\(A\) and scalar cone \(B\).  Using the common-congruence covariance of
eq.~\eqref{eq:jhep-network-flow}, maintain the fermion frame at
\(A=\Id\).  The relative matrix
\begin{equation}
R=A^{-1/2}BA^{-1/2}
\label{eq:relative-R-sec3}
\end{equation}
then obeys
\begin{equation}
\dot R
=
\alpha(\Id-R)
-
\beta\{\cH(R),R\}.
\label{eq:relative-flow}
\end{equation}
Because \(\cH(R)\) is a spectral function of \(R\),
eq.~\eqref{eq:relative-flow} is diagonal in the instantaneous eigenbasis of
\(R\).  If \(r_i>0\) is an eigenvalue of \(R\), define
\begin{equation}
m_i(y)=1-y+yr_i.
\end{equation}
The corresponding eigenvalue \(h_i\) of \(\cH(R)\) is
\begin{equation}
h_i
=
(r_i-1)
\int_0^1\dd y\,
\frac{y(1-y)}
{m_i(y)\sqrt{\prod_k m_k(y)}}.
\label{eq:hi-sign}
\end{equation}
Hence
\begin{equation}
\operatorname{sgn}h_i
=
\operatorname{sgn}(r_i-1),
\label{eq:hi-sign-result}
\end{equation}
and
\begin{equation}
\dot r_i
=
\alpha(1-r_i)-2\beta r_i h_i
\label{eq:ri-flow}
\end{equation}
implies
\begin{equation}
\operatorname{sgn}\dot r_i
=
-\operatorname{sgn}(r_i-1).
\label{eq:ri-monotone}
\end{equation}
Thus every relative eigenvalue moves monotonically toward unity throughout
the finite SPD domain.

A convenient one-edge Lyapunov function is
\begin{equation}
\mathscr L(R)
=
\frac12\Tr[(\log R)^2]
=
\frac12\sum_i(\log r_i)^2.
\label{eq:single-edge-Lyapunov}
\end{equation}
Along eq.~\eqref{eq:relative-flow},
\begin{align}
\dot{\mathscr L}
={}&
\alpha
\sum_i
\log r_i
\left(
\frac1{r_i}-1
\right)
-
2\beta
\sum_i
(\log r_i)h_i
\nonumber\\
<&\,0
\qquad
(R\neq\Id).
\label{eq:single-edge-Ldot}
\end{align}
The strict sign follows from
eq.~\eqref{eq:hi-sign-result}.  A single interacting pair therefore locks
globally at finite mismatch.  For a network, the extremal backreaction of
shared nodes is captured more naturally by a metric on the full positive
cone.

\subsection{Thompson geometry for network mismatch}
\label{subsec:thompson}

For \(A,B\in\SPD(3)\), let \(\lambda_{\max}(X)\) denote the largest
eigenvalue of a real symmetric matrix \(X\).  Thompson's
metric~\cite{Thompson1963} is
\begin{equation}
\begin{aligned}
d_T(A,B)=\max\bigl\{&
\log\lambda_{\max}(A^{-1/2}BA^{-1/2}),\\
&\log\lambda_{\max}(B^{-1/2}AB^{-1/2})
\bigr\}.
\end{aligned}
\label{eq:Thompson-distance}
\end{equation}
It is invariant under a common congruence and retains relative isotropic
speed mismatch:
\begin{equation}
d_T(G,cG)=|\log c|.
\label{eq:Thompson-scale}
\end{equation}
The latter property distinguishes it from projective metrics that quotient
out relative scalar rescalings.  Thompson and Hilbert geometries on positive cones have a well-developed
contraction theory~\cite{Birkhoff1957,Thompson1963}, including matrix-valued
and noncommutative consensus settings
\cite{SepulchreSarletteRouchon2010,GaubertQu2014}.  Here the metric supplies
the geometry; the nontrivial step is to show that the QFT-derived Yukawa
vector field of eq.~\eqref{eq:jhep-network-flow} obeys the required extremal
inequality globally at finite mismatch.

For a finite network define
\begin{equation}
\DT
=
\max_{i,j}d_T(G_i,G_j)
=
\log\Lambda,
\label{eq:network-diameter}
\end{equation}
where ordered pairs allow
\begin{equation}
\Lambda
=
\max_{i,j}
\lambda_{\max}
\left(
G_i^{-1/2}G_jG_i^{-1/2}
\right).
\label{eq:Lambda-def}
\end{equation}
Equivalently, \(\Lambda\) is the least number satisfying
\begin{equation}
G_j\preceq \Lambda G_i
\qquad
\text{for every ordered pair }(i,j).
\label{eq:global-Loewner-bounds}
\end{equation}
This order-theoretic form makes the edge signs transparent.

\subsection{Yukawa kernel identities and extremal rigidity}
\label{subsec:kernel-rigidity}

The global argument rests on two identities specific to the finite-matrix
Yukawa kernel.

\paragraph{Inversion identity.}
Changing \(y\mapsto1-y\) in
eq.~\eqref{eq:jhep-H-kernel} and using
\begin{equation}
M_y(R^{-1})
=
R^{-1}[y\Id+(1-y)R]
\end{equation}
gives
\begin{equation}
\cH(R^{-1})
=
-\sqrt{\det R}\,\cH(R).
\label{eq:H-inversion}
\end{equation}
This relation maps the lower-extremum sign problem to its upper-extremum
dual.

\paragraph{Rank-one extremal sign.}
Let \(u\) be a Euclidean unit vector and suppose
\begin{equation}
R\succeq uu^T.
\label{eq:R-rankone-bound}
\end{equation}
The Schur-complement characterization of
eq.~\eqref{eq:R-rankone-bound} gives
\begin{equation}
u^TR^{-1}u\le1.
\label{eq:inverse-expectation}
\end{equation}
Since matrix inversion is operator convex on \(\SPD(3)\),
\begin{equation}
M_y(R)^{-1}
\preceq
(1-y)\Id+yR^{-1}.
\label{eq:inverse-convex}
\end{equation}
Therefore
\begin{equation}
u^T[\Id-M_y(R)^{-1}]u
\ge
y[1-u^TR^{-1}u]
\ge0,
\end{equation}
and the positive weight in
eq.~\eqref{eq:jhep-H-kernel} yields
\begin{equation}
R\succeq uu^T
\quad\Longrightarrow\quad
u^T\cH(R)u\ge0.
\label{eq:rankone-sign}
\end{equation}

The equality case is rigid.  If
\(u^T\cH(R)u=0\), the continuous nonnegative integrand must vanish.
Its expansion at \(y=0\) gives \(u^TRu=1\).  Together with
\(R-uu^T\succeq0\), this implies
\((R-uu^T)u=0\), and hence
\begin{equation}
u^T\cH(R)u=0
\quad\Longleftrightarrow\quad
Ru=u
\qquad
(R\succeq uu^T).
\label{eq:equality-rigidity}
\end{equation}
Using eq.~\eqref{eq:H-inversion}, the dual statement is
\begin{equation}
R^{-1}\succeq uu^T
\quad\Longrightarrow\quad
u^T\cH(R)u\le0,
\label{eq:dual-rankone-sign}
\end{equation}
with equality again if and only if \(Ru=u\).  The sign controls
nonexpansion; the equality condition will later transmit strictness through
a connected graph.

\subsection{Global nonexpansion}
\label{subsec:global-nonexpansion}

For a real-valued function \(f(s)\), we write \(D^+f(s)\) for its
upper-right Dini derivative.

We now apply the preceding identities to an active extremal pair.  Let
\((i,j)\) realize the maximum in eq.~\eqref{eq:Lambda-def}, and choose a
generalized eigenvector \(x\) normalized by
\begin{equation}
G_jx=\Lambda G_ix,
\qquad
x^TG_ix=1.
\label{eq:active-generalized-eigenpair}
\end{equation}
Then \(x^TG_jx=\Lambda\).  We first establish the endpoint signs
\begin{equation}
x^T\dot G_i x\ge0,
\qquad
x^T\dot G_j x\le0.
\label{eq:endpoint-signs}
\end{equation}

If the lower endpoint is bosonic, \(G_i=B_r\), the global bound
\(G_j\preceq\Lambda A_a\) holds for every neighbor \(a\sim r\).
Evaluating it on \(x\) gives \(x^TA_ax\ge1=x^TB_rx\), and
eq.~\eqref{eq:jhep-network-B} therefore gives
\(x^T\dot B_rx\ge0\).

If the lower endpoint is fermionic, \(G_i=A_a\), define for each incident
scalar edge
\begin{equation}
\begin{aligned}
u&=A_a^{1/2}x,
& C&=A_a^{-1/2}G_jA_a^{-1/2},\\
R&=A_a^{-1/2}B_rA_a^{-1/2}.
\end{aligned}
\label{eq:lower-fermion-frames}
\end{equation}
Then \(u\) is unit and \(Cu=\Lambda u\).  Global extremality applied to
the ordered pair \((B_r,G_j)\) gives
\(G_j\preceq\Lambda B_r\), hence
\(C\preceq\Lambda R\).  Since \(C/\Lambda\) is positive semidefinite and
has eigenvalue one along \(u\),
\begin{equation}
R\succeq C/\Lambda\succeq uu^T.
\end{equation}
Equation~\eqref{eq:rankone-sign} then gives
\(u^T\cH(R)u\ge0\).  Every incident term in
eq.~\eqref{eq:jhep-network-A} is therefore nonnegative on \(x\), proving
the lower fermionic part of eq.~\eqref{eq:endpoint-signs}.

The upper endpoint is analogous.  If \(G_j=B_r\), the bound
\(A_a\preceq\Lambda G_i\) for every \(a\sim r\) implies
\(x^TA_ax\le\Lambda=x^TB_rx\), and hence
\(x^T\dot B_rx\le0\).  If \(G_j=A_a\), define
\begin{equation}
\begin{aligned}
v&=\frac{A_a^{1/2}x}{\sqrt{\Lambda}},
& C&=A_a^{-1/2}G_iA_a^{-1/2},\\
R&=A_a^{-1/2}B_rA_a^{-1/2}.
\end{aligned}
\label{eq:upper-fermion-frames}
\end{equation}
Then \(v\) is unit and \(Cv=\Lambda^{-1}v\).  The global order bound
\(B_r\preceq\Lambda G_i\) becomes
\(R\preceq\Lambda C\), so
\begin{equation}
R^{-1}
\succeq
\Lambda^{-1}C^{-1}
\succeq vv^T.
\end{equation}
The dual sign
eq.~\eqref{eq:dual-rankone-sign} gives
\(v^T\cH(R)v\le0\), and every incident fermion-edge contribution is
nonpositive at the upper endpoint.  This proves
eq.~\eqref{eq:endpoint-signs} for all endpoint types.

To include degenerate active generalized eigenvalues, define
\begin{equation}
\Lambda_{ij}
=
\lambda_{\max}
\left(
G_i^{-1/2}G_jG_i^{-1/2}
\right)
\end{equation}
and its active generalized eigenspace
\begin{equation}
\mathcal E_{ij}
=
\{x:\;G_jx=\Lambda_{ij}G_ix\}.
\end{equation}
The perturbation formula for a symmetric matrix pencil gives the upper Dini
derivative
\begin{equation}
\Dini\Lambda_{ij}
=
\max_{\substack{x\in\mathcal E_{ij}\\x^TG_ix=1}}
x^T
\left(
\dot G_j-\Lambda_{ij}\dot G_i
\right)x.
\label{eq:pencil-Dini}
\end{equation}
The endpoint signs established above hold for every vector in an active
eigenspace, so
\begin{equation}
\Dini\Lambda_{ij}\le0
\end{equation}
for every active ordered pair.  If \(\mathcal A\) denotes the finite set of
ordered pairs attaining the network maximum \(\Lambda\), then
\begin{equation}
\Dini\Lambda
\le
\max_{(i,j)\in\mathcal A}
\Dini\Lambda_{ij}
\le0.
\label{eq:Lambda-nonexpansion}
\end{equation}

\begin{theorem}[Global nonexpansion]
\label{thm:global-nonexpansion}
Let the one-loop network flow
eq.~\eqref{eq:jhep-network-flow} exist in the positive-definite domain, with
positive edge rates on the fixed Yukawa graph.  Then its Thompson diameter
satisfies
\begin{equation}
\Dini\DT\le0.
\label{eq:DT-nonexpansion}
\end{equation}
Hence the relative cone diameter is nonincreasing for arbitrary finite SPD
initial mismatch throughout the controlled flow.
\end{theorem}

\begin{proof}
Equation~\eqref{eq:network-diameter} gives
\(\DT=\log\Lambda\).  Since \(\Lambda>0\),
eq.~\eqref{eq:Lambda-nonexpansion} implies
\[
\Dini\DT
=
\frac{\Dini\Lambda}{\Lambda}
\le0.
\]
\end{proof}

The theorem also gives a useful compactness statement for the normalized
relative sector.  Let \(D_0=\DT(0)\) and choose any reference cone as the
identity by a common congruence.  Every normalized cone \(\widetilde G_i\)
then satisfies
\begin{equation}
e^{-D_0}\Id
\preceq
\widetilde G_i
\preceq
e^{D_0}\Id.
\label{eq:compact-Loewner-sector}
\end{equation}
Thus the initial Thompson diameter controls a compact Loewner interval for
all relative propagation matrices.

For a complete bipartite graph, the equality rigidity in
eq.~\eqref{eq:equality-rigidity} already makes the active diameter strictly
inward at every nonsynchronized state.  Sparse connected graphs are subtler:
an active extremum may be shielded by several edges on which the first
derivative vanishes.  Theorem~\ref{thm:global-nonexpansion} therefore
establishes global control of the relative spread, while the next section
develops the finite-time propagation mechanism by which strict contraction
travels through a general connected interaction graph.

\section{From nonexpansion to global common-cone locking}
\label{sec:finite-window-locking}

Theorem~\ref{thm:global-nonexpansion} shows that the largest relative
cone mismatch can never increase.  On a sparse network, locking requires a
second step: an extremal cone can be separated from the first genuinely
mismatched edge by several equal-extremum edges, so its instantaneous
Thompson derivative may vanish before strictness reaches that endpoint.  The
central question is therefore whether strictness propagates through a
connected Yukawa graph in finite accumulated interaction time.  We establish
this under uniform perturbative control and persistent edge activity.

\subsection{Delayed strictness on sparse graphs}
\label{subsec:why-window}

The mechanism is easiest to see on a path.  Suppose several neighboring
nodes initially share the same extremal generalized eigenvalue along an
active direction, while the first mismatch occurs only farther along the
graph.  The endpoint sign argument of
section~\ref{subsec:global-nonexpansion} then saturates on the first edge:
the relevant rank-one inequality is an equality and the first derivative of
the active diameter vanishes.  The next node can nevertheless begin to move
because of a mismatch on its other side, after which strictness is
transmitted one edge closer to the endpoint.  Repeating this mechanism
produces a finite-order delay rather than a failure of contraction.

This motivates a finite-window formulation.  We separate the overall
strength of the running Yukawa interactions from their relative edge
weights by introducing an accumulated interaction clock.  The theorem below
is stated entirely in that clock; section~\ref{subsec:qft-persistence} will
later give an explicit field-theoretic realization for which the accumulated
clock diverges in the infrared.

\subsection{Persistent-interaction assumptions and the locking theorem}
\label{subsec:persistent-assumptions}

Let \(Z\) denote the remaining running couplings and wave-function data on
which the positive cone-flow rates may depend.  We assume that the extended
one-loop RG state
\begin{equation}
\mathsf X=(\{G_i\},Z)
\label{eq:extended-state}
\end{equation}
evolves autonomously and real analytically on a compact invariant
perturbative domain \(\mathcal K\).  Let
\(\chi(\mathsf X)>0\) be a common interaction scale and define
\begin{equation}
\frac{\dd s}{\dd\ell}
=
\chi(\mathsf X(\ell)).
\label{eq:interaction-clock}
\end{equation}
The normalized edge rates are
\begin{equation}
\widehat\alpha_{ar}
=
\frac{\alpha_{ar}}{\chi},
\qquad
\widehat\beta_{ar}
=
\frac{\beta_{ar}}{\chi}.
\label{eq:normalized-rates}
\end{equation}
The persistence condition is that, on the fixed connected graph,
\begin{equation}
0<m
\le
\widehat\alpha_{ar},
\widehat\beta_{ar}
\le
M<\infty
\qquad
\text{for every }(a,r)\in E,
\label{eq:uniform-rate-bounds}
\end{equation}
with constants \(m,M\) independent of \(s\).  The four assumptions have
distinct roles: connectedness permits strictness to propagate across the
network, eq.~\eqref{eq:uniform-rate-bounds} prevents active edges from
disappearing in the normalized clock, compactness supplies uniform bounds,
and analyticity controls finite-order propagation through temporarily
saturated extremal branches.

In \(s\)-time the cone equations retain the form
\begin{subequations}
\label{eq:normalized-network-flow}
\begin{align}
\frac{\dd A_a}{\dd s}
&=
2\sum_{r\sim a}
\widehat\beta_{ar}\,
\cQ(A_a,B_r),
\label{eq:normalized-A}\\
\frac{\dd B_r}{\dd s}
&=
\sum_{a\sim r}
\widehat\alpha_{ar}\,
(A_a-B_r).
\label{eq:normalized-B}
\end{align}
\end{subequations}

\begin{theorem}[Finite-window contraction]
\label{thm:finite-window}
Consider the normalized one-loop flow
eq.~\eqref{eq:normalized-network-flow} on a fixed finite connected Yukawa
graph.  Suppose the extended state remains in a compact invariant
perturbative domain, the normalized vector field is real analytic there,
and the edge rates obey eq.~\eqref{eq:uniform-rate-bounds}.  Then for every
\(T>0\) there exists
\begin{equation}
0<\rho_T<1
\label{eq:rhoT}
\end{equation}
such that every trajectory in the domain satisfies
\begin{equation}
\DT(s+T)
\le
\rho_T\,\DT(s)
\label{eq:finite-window-contraction}
\end{equation}
whenever both times belong to the controlled flow interval.  Consequently,
if \(s(\ell)\to\infty\) along the infrared RG trajectory, then
\begin{equation}
\DT(\ell)\longrightarrow0.
\label{eq:DT-locking}
\end{equation}
Thus all sectors asymptotically share one relative propagation geometry.
\end{theorem}

The proof has two regimes.  Away from synchronization, the equality
rigidity of section~\ref{subsec:kernel-rigidity} rules out any
nonsynchronized interval of constant diameter and compactness upgrades this
pointwise strictness to a uniform finite-window factor.  Near
synchronization, the Yukawa kernel linearizes to an irreducible
bidirectional consensus system, which supplies a second uniform factor.
The remainder of this section gives the proof architecture; the branch and
compactness technicalities are collected in
appendix~\ref{app:global-proofs}.

\subsection{Strong endpoint propagation}
\label{subsec:endpoint-propagation}

We first record the mechanism that turns equality at one extremal edge into
information about the next edge.  Let
\((i,j)\) be an active ordered pair with
\begin{equation}
G_jx=\Lambda G_ix,
\qquad
x^TG_ix=1,
\qquad
\Lambda=e^{\DT}>1.
\label{eq:active-pair-sec4}
\end{equation}
Suppose the lower endpoint derivative along \(x\) vanishes.  Every incident
term that entered the nonnegative sum in the proof of
Theorem~\ref{thm:global-nonexpansion} must then vanish separately.

For a bosonic lower endpoint \(B_r\), this gives
\begin{equation}
x^T(A_a-B_r)x=0
\qquad
\text{for every }a\sim r.
\label{eq:boson-equality-edge}
\end{equation}
The global Loewner ordering used in section~\ref{subsec:global-nonexpansion}
makes \(A_a-B_r\) positive semidefinite on the active extremal direction,
so the equality implies the corresponding edge agrees exactly along that
direction.  For a fermionic lower endpoint, the equality condition
eq.~\eqref{eq:equality-rigidity} gives, in the fermion frame,
\begin{equation}
R\,u=u,
\label{eq:fermion-equality-edge}
\end{equation}
which is the same directional matching statement in invariant form.  The
upper endpoint has the dual property.

The next derivative detects the first mismatch behind such a saturated
edge.  Let \(D\cH_R[K]\) denote the Fr\'echet derivative of the Yukawa
kernel at \(R\) in the symmetric matrix direction \(K\).  If \(Rz=z\),
then for every such perturbation \(K\) generated by the neighboring cone
flow,
\begin{equation}
z^T(D\cH_R[K])z
=
c(R)\,z^TKz,
\qquad
c(R)>0.
\label{eq:Frechet-transfer}
\end{equation}
The Fr\'echet derivative and the explicit positive integral for
\(c(R)\) are derived in appendix~\ref{app:selfenergies};
appendix~\ref{app:global-proofs} uses this identity in the full
endpoint-propagation induction.  Equation~\eqref{eq:Frechet-transfer}
has two consequences.  First, a strict inward coefficient generated at a
mismatched edge retains its sign when transmitted across an equality edge.
Second, crossing one saturated edge delays the first nonzero inward term by
one derivative order.

Iterating along a path gives the strong endpoint-propagation statement:

\begin{lemma}[Endpoint propagation]
\label{lem:endpoint-propagation}
Consider an analytic trajectory of
eq.~\eqref{eq:normalized-network-flow} and an analytic active generalized
eigenvalue branch with \(\Lambda>1\).  If its first derivative vanishes,
then either the branch is already synchronized along the entire connected
graph or there exists a finite derivative order at which the branch is
strictly inward.  The first nonzero coefficient has the contraction sign.
\end{lemma}

The induction alternates between the bosonic linear edge
eq.~\eqref{eq:normalized-B} and the fermionic transfer
eq.~\eqref{eq:Frechet-transfer}.  Derivatives of the normalized rates and of
the moving fermion frame multiply edge fluxes whose lower-order coefficients
have already vanished, so they do not change the first nonzero order or its
sign.  Since the graph is finite and connected, a nonsynchronized active
endpoint encounters a mismatch after finitely many propagation steps.
Appendix~\ref{app:global-proofs} gives the full Taylor-order induction.

\subsection{Strictness away from synchronization}
\label{subsec:away-from-sync}

The Thompson diameter is the maximum of finitely many largest generalized
eigenvalues and can lose analyticity when active branches switch.  The
matrices in the extended RG state remain analytic in \(s\), and
the generalized eigenvalues of each symmetric matrix pencil admit local
analytic parametrizations.  Hence on every compact time interval the
diameter is the upper envelope of finitely many analytic generalized
eigenvalue branches, after a finite local subdivision.

Suppose a nonsynchronized trajectory had
\begin{equation}
\DT(s)=D_*>0
\label{eq:constant-diameter}
\end{equation}
through a nontrivial time interval.  The finite family of analytic active
branches covers that interval.  A standard finite-branch/Baire argument
then yields a subinterval on which one analytic branch supports the
constant envelope.  All of its derivatives vanish there, contradicting
Lemma~\ref{lem:endpoint-propagation}.  Therefore every nonsynchronized
trajectory satisfies
\begin{equation}
\DT(s+T)<\DT(s)
\qquad
\text{for every }T>0
\label{eq:pointwise-window-strict}
\end{equation}
as long as the trajectory remains in the controlled domain.

Pointwise strictness becomes uniform away from the synchronized manifold.
For \(\varepsilon>0\), define
\begin{equation}
\mathcal K_\varepsilon
=
\left\{
\mathsf X\in\mathcal K:
\DT(\mathsf X)\ge\varepsilon
\right\}.
\label{eq:Kepsilon}
\end{equation}
This set is compact.  Let \(\Phi_T\) be the \(s\)-time-\(T\) flow map and
define
\begin{equation}
R_T(\mathsf X)
=
\frac{\DT(\Phi_T\mathsf X)}
{\DT(\mathsf X)}.
\label{eq:RT-ratio}
\end{equation}
Continuity of the flow and of Thompson's metric makes \(R_T\) continuous
on \(\mathcal K_\varepsilon\), while
eq.~\eqref{eq:pointwise-window-strict} gives \(R_T<1\) pointwise.
Therefore
\begin{equation}
\rho_{T,\varepsilon}^{\rm out}
=
\max_{\mathsf X\in\mathcal K_\varepsilon}
R_T(\mathsf X)
<1.
\label{eq:rho-out}
\end{equation}

\subsection{Near-synchronization consensus contraction}
\label{subsec:near-sync}

Near synchronization, the ratio eq.~\eqref{eq:RT-ratio} approaches the
synchronized manifold, where compactness alone does not yield a strict
factor.  A common congruence places the synchronized cone at the identity.  Write
\begin{equation}
A_a=\Id+X_a,
\qquad
B_r=\Id+Y_r,
\qquad
\|X_a\|,\|Y_r\|=O(\varepsilon).
\label{eq:near-sync-variables}
\end{equation}
The finite-matrix kernel has the Fr\'echet expansion
\begin{equation}
\cH(\Id+E)
=
\frac16E
+
O(\|E\|^2),
\label{eq:H-linearization}
\end{equation}
so
\begin{equation}
\cQ(\Id+X_a,\Id+Y_r)
=
\frac16(Y_r-X_a)
+
O(\varepsilon^2).
\label{eq:Q-linearization}
\end{equation}
To first order, eq.~\eqref{eq:normalized-network-flow} becomes
\begin{subequations}
\label{eq:linear-consensus}
\begin{align}
\dot X_a
&=
\sum_{r\sim a}
\frac{\widehat\beta_{ar}}{3}
(Y_r-X_a),
\label{eq:linear-consensus-F}\\
\dot Y_r
&=
\sum_{a\sim r}
\widehat\alpha_{ar}
(X_a-Y_r).
\label{eq:linear-consensus-B}
\end{align}
\end{subequations}
Here and below a dot denotes \(\dd/\dd s\).

For each independent entry of a symmetric perturbation matrix,
eq.~\eqref{eq:linear-consensus} is an ordinary continuous-time consensus
system on the bidirected version of the Yukawa graph.  Its generator is
Metzler with zero row sums.  Connectedness together with
eq.~\eqref{eq:uniform-rate-bounds} makes the transition operator over every
fixed \(T>0\) uniformly scrambling: there exists
\begin{equation}
0<\eta_T<1
\label{eq:etaT}
\end{equation}
depending only on \(T\), the finite graph and the bounds \(m,M\), such that
the scalar oscillation contracts by at most \(\eta_T\).  Equivalently, the
Dobrushin coefficient of the transition operator is uniformly below one.
A pathwise variation-of-constants estimate gives a uniform positive
lower bound on every entry of the time-\(T\) transition matrix; the explicit
bound and the resulting Dobrushin estimate are given in
appendix~\ref{app:global-proofs}.

Near the synchronized manifold, Thompson distance is locally equivalent to
the oscillation of the logarithmic matrix perturbations.  More precisely,
for sufficiently small \(\varepsilon\) there are uniform constants
\(c_-,c_+>0\) such that
\begin{equation}
c_-\,\operatorname{osc}(X,Y)
\le
\DT
\le
c_+\,\operatorname{osc}(X,Y),
\label{eq:local-Thompson-equivalence}
\end{equation}
where \(\operatorname{osc}\) may be taken as the operator-norm diameter
\(\max_{u,v}\|Z_u-Z_v\|_{\rm op}\) on the quotient by the common matrix
mode.  The quadratic remainder in eq.~\eqref{eq:Q-linearization} is uniformly
\(O(\varepsilon)\) relative to the linear spread.  For sufficiently small
\(\varepsilon\), it preserves a strict finite-window factor:
\begin{equation}
\DT(s+T)
\le
\rho_T^{\rm in}\,\DT(s),
\qquad
0<\rho_T^{\rm in}<1,
\qquad
\DT(s)<\varepsilon.
\label{eq:rho-in}
\end{equation}

\subsection{Completion of the finite-window theorem}
\label{subsec:complete-window}

Choose \(\varepsilon>0\) small enough for
eq.~\eqref{eq:rho-in}.  On the complementary region
\(\DT\ge\varepsilon\), use
eq.~\eqref{eq:rho-out}.  Then
\begin{equation}
\rho_T
=
\max
\left\{
\rho_{T,\varepsilon}^{\rm out},
\rho_T^{\rm in}
\right\}
<1
\label{eq:rho-patched}
\end{equation}
proves eq.~\eqref{eq:finite-window-contraction} for the entire compact
controlled domain.

Iterating over windows gives
\begin{equation}
\DT(s+nT)
\le
\rho_T^n\DT(s),
\qquad
n=0,1,2,\ldots.
\label{eq:window-iteration}
\end{equation}
Thus any infrared trajectory for which the accumulated interaction time
diverges satisfies \(\DT\to0\).  The theorem establishes asymptotic
synchronization of the relative propagation geometries.  Section~\ref{subsec:qft-persistence} shows that its
persistence assumptions are dynamically realized on an open set of
finite-mismatch Yukawa initial data.

The same argument applies component by component when the interaction graph
is disconnected.  If the graph has \(k\) connected components, the
Thompson diameter within each component tends to zero under the same
componentwise persistence conditions.  Hence the infrared theory contains at most \(k\) mutually unlocked
propagation-cone sectors.  Distinct components can still converge to the same
cone for particular initial data.

\subsection{Representative fixed-rate dynamics}
\label{subsec:fixed-rate-numerics}

Before returning to the running Yukawa couplings, we illustrate the
finite-window theorem directly in the normalized cone flow with fixed
positive edge rates.  Figure~\ref{fig:global-locking} shows two
representative evolutions.  In panel~(a), the same four initially anisotropic
\(3\times3\) SPD cones are evolved with
\begin{equation}
\widehat\alpha_e=\widehat\beta_e=1
\label{eq:fixed-unit-rates}
\end{equation}
on two connected bipartite topologies: the alternating path \(P_4\) and the
complete graph \(K_{2,2}\).  The Thompson diameter decreases in both cases.
For this common initial condition, the complete graph contracts faster over
the displayed interval because each sector receives strict information from
more incident edges.  This example illustrates how additional incident edges
can accelerate contraction for a fixed normalization and initial condition.

Panel~(b) shows sixteen additional connected networks with two fermionic and
three bosonic nodes.  The edge rates are independently chosen within a fixed
positive interval,
\begin{equation}
e^{-0.35}
\le
\widehat\alpha_e,\widehat\beta_e
\le
e^{0.35},
\label{eq:random-rate-window}
\end{equation}
and the initial cones are random nonaligned SPD matrices.  Each thin curve is
the normalized diameter of one realization and the thick curve is their
pointwise median.  These integrations provide a direct numerical illustration
of the finite-mismatch dynamics and of the monotonicity established by
Theorems~\ref{thm:global-nonexpansion} and~\ref{thm:finite-window}.
Numerical quadrature, ODE tolerances and the random seeds defining the
displayed realizations are given in
appendix~\ref{appsub:global-numerics}.

\begin{figure*}[t]
\centering
\includegraphics[width=0.96\textwidth]{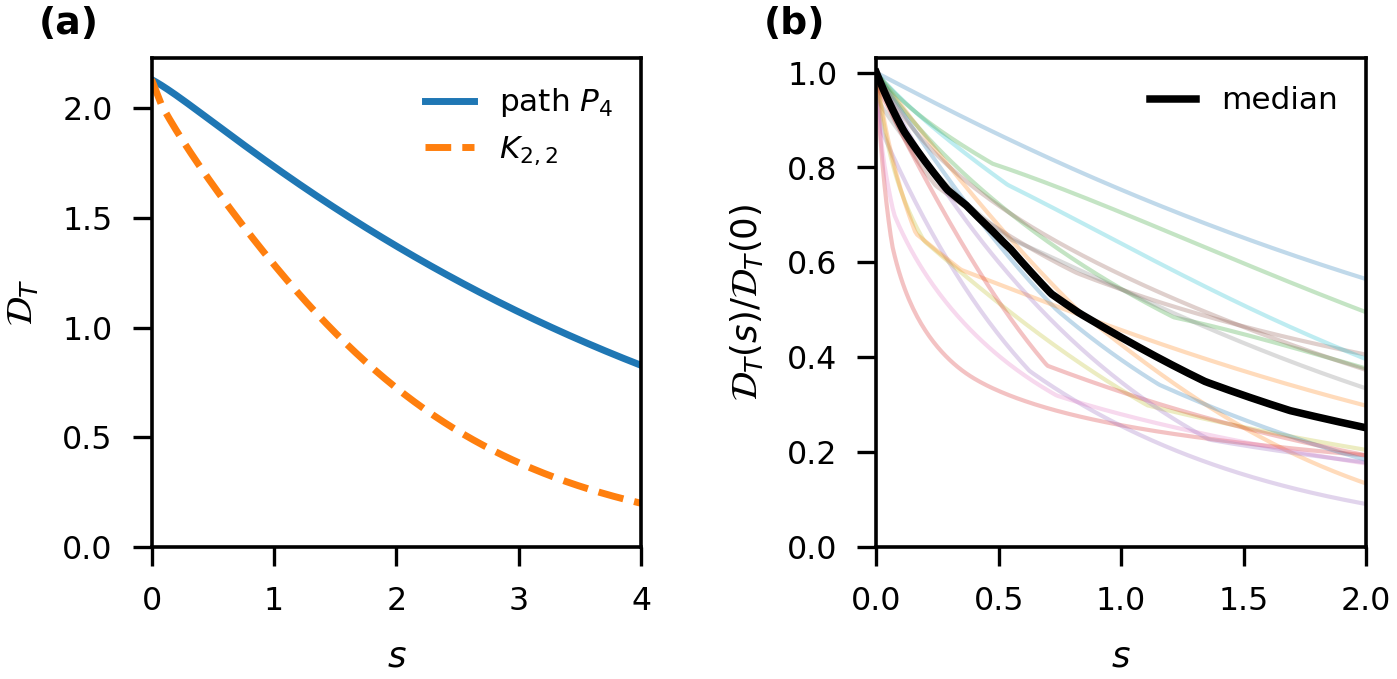}
\caption{
Representative locking trajectories of the normalized finite-matrix flow.
(a) Thompson diameter \(\mathcal D_T\) for the same four initial SPD cones
on the path \(P_4\) and the complete bipartite graph \(K_{2,2}\), with
\(\widehat\alpha_e=\widehat\beta_e=1\).
(b) Normalized diameter for sixteen random connected \(2+3\) bipartite
networks with positive rates in the interval
eq.~\eqref{eq:random-rate-window}; the black curve is the pointwise median.
The horizontal variable \(s\) is accumulated interaction time.
}
\label{fig:global-locking}
\end{figure*}

\subsection{Field-theoretic realization of persistent activity}
\label{subsec:qft-persistence}

Theorem~\ref{thm:finite-window} isolates the condition required from the
remaining coupling flow: the accumulated interaction clock must diverge and
every active edge must retain a finite fraction of that common clock.  We
now show that this condition is generated dynamically in an explicit
two-channel Dirac Yukawa theory on an open set of finite-mismatch initial
data.

Consider one Dirac node \(F\) carrying the Dirac field \(\psi_F\), two
scalar nodes \(B_1,B_2\), and a two-dimensional internal flavor space.
Choose
\begin{equation}
T_1=\sigma_x,
\qquad
T_2=\sigma_z.
\label{eq:pauli-two-channel}
\end{equation}
The matrices are Hermitian involutions, Hilbert--Schmidt orthogonal, and
anticommute.  Hence the Pauli-channel closure discussed in
section~\ref{subsec:channel-closure} applies.  Let \(a,b>0\) denote the
Weyl-block amplitudes defined in appendix~\ref{app:internal-algebra},
equivalently a common rescaling of the two dimensionless Dirac Yukawa
matrices.

The transformations
\(\phi_1\mapsto-\phi_1,\ \psi_F\mapsto T_2\psi_F\) and
\(\phi_2\mapsto-\phi_2,\ \psi_F\mapsto T_1\psi_F\) generate a
\(\mathbb Z_2\times\mathbb Z_2\) symmetry.  After tuning the two relevant
scalar masses to the critical surface, the complete renormalizable quartic
potential compatible with this symmetry is
\begin{equation}
V_4(\phi_1,\phi_2)
=
\frac{\lambda_1}{4!}\phi_1^4
+
\frac{\lambda_2}{4!}\phi_2^4
+
\frac{\lambda_{12}}{4}\phi_1^2\phi_2^2,
\label{eq:two-channel-quartic-potential}
\end{equation}
which defines the three dimensionless quartic couplings
\(\lambda_1,\lambda_2,\lambda_{12}\) in the same \(4\pi\)-rescaled
convention as the Yukawa couplings.  The mass parameters are set to zero on
the critical surface; in a mass-independent scheme they do not enter the
one-loop beta functions of the dimensionless kinetic, Yukawa, and quartic
couplings.

At a common propagation cone, define
\(\beta_a=\mu\,\dd a/\dd\mu\) and
\(\beta_b=\mu\,\dd b/\dd\mu\), where \(\mu\) is the renormalization scale.
The one-loop matrix Yukawa beta function
specialized to this four-component Dirac realization gives, in
\(4-\varepsilon\) dimensions,
\begin{subequations}
\label{eq:two-channel-beta-epsilon}
\begin{align}
\beta_a
&=
-\frac{\varepsilon}{2}a
+
a(7a^2-b^2),
\label{eq:beta-a-epsilon}\\
\beta_b
&=
-\frac{\varepsilon}{2}b
+
b(7b^2-a^2).
\label{eq:beta-b-epsilon}
\end{align}
\end{subequations}
Appendix~\ref{app:internal-algebra} derives the coefficients directly from
the general one-loop matrix beta function
\cite{PannellStergiou2023}, including the Dirac-to-Weyl map.  For comparison
with the standard Nambu--Jona-Lasinio--Yukawa (NJLY) normalization, the
equal-coupling ray \(a=b=g\) gives
\begin{equation}
\beta_g
=
-\frac{\varepsilon}{2}g+6g^3,
\qquad
g_*^2=\frac{\varepsilon}{12},
\label{eq:NJLY-matching}
\end{equation}
which coincides, after restoring the conventional overall
\(\varepsilon\) scaling, with the NJLY one-loop fixed-point value for
\(N_D=2\), eq.~(3.10) of Ref.~\cite{PannellStergiou2023}.  Here the single
graph node \(F\) carries a two-dimensional internal Dirac flavor space,
corresponding to two Dirac flavors in the standard NJLY counting.  This
agreement is a normalization check: the coefficients \(7\) and \(-1\) are
derived independently from the general Weyl beta function in
appendix~\ref{app:internal-algebra}.  That appendix also shows explicitly
why suppressing the second scalar channel gives the one-scalar GNY
denominator \(N+6\), whereas the anticommuting two-channel equal ray gives
the NJLY denominator \(N+4\).

For the four-dimensional marginal theory studied here, set
\(\varepsilon=0\) and orient the RG variable \(\ell\) toward the infrared.
A dot denotes \(\dd/\dd\ell\).  Up to a common positive normalization
\(\kappa\), eqs.~\eqref{eq:two-channel-beta-epsilon} become
\begin{subequations}
\label{eq:two-channel-IR}
\begin{align}
\dot a
&=
-\kappa a(7a^2-b^2),
\label{eq:a-IR}\\
\dot b
&=
-\kappa b(7b^2-a^2),
\label{eq:b-IR}
\end{align}
\end{subequations}
at the common cone.  The multiplicative form of the equations already shows
that the positive quadrant is invariant: positive initial couplings remain
positive for every finite RG time for which the perturbative flow exists.

Define
\begin{equation}
S=a^2+b^2,
\qquad
q=\frac{a^2-b^2}{a^2+b^2}.
\label{eq:S-q-def}
\end{equation}
Because \(a,b>0\), one has \(S>0\).  Define the normalized quartic ratios
\begin{equation}
u=\frac{\lambda_1}{S},
\qquad
v=\frac{\lambda_2}{S},
\qquad
w=\frac{\lambda_{12}}{S}.
\label{eq:normalized-quartics}
\end{equation}
Appendix~\ref{appsub:quartic-completion} derives their one-loop flow.  At a
common cone, the coupled normalized system for \((q,u,v,w)\) has the fixed
point
\begin{equation}
(q_*,u_*,v_*,w_*)
=
\left(0,\frac{12}{5},\frac{12}{5},\frac45\right),
\label{eq:quartic-normalized-fixed-point}
\end{equation}
corresponding to the \(O(2)\)-symmetric quartic ray
\(\lambda_1=\lambda_2=3\lambda\),
\(\lambda_{12}=\lambda\), with \(\lambda/g^2=8/5\).
The four linearized eigenvalues in the infrared-oriented normalized clock are
\begin{equation}
-8,\qquad -18,\qquad -\frac{58}{5},\qquad -\frac{82}{5},
\label{eq:quartic-hyperbolic-eigenvalues}
\end{equation}
so this point is a hyperbolic sink, meaning that all linearized directions
are strictly infrared attractive.  Appendix~\ref{appsub:quartic-trapping}
shows that the attraction persists for sufficiently small finite cone
mismatch and produces a compact forward-invariant neighborhood
\(\mathcal U_*\) of \(p_*\) in the \((q,u,v,w)\) variables.

A direct calculation from
eq.~\eqref{eq:two-channel-IR} gives
\begin{subequations}
\label{eq:S-q-flow}
\begin{align}
\dot S
&=
-\kappa S^2(6+8q^2),
\label{eq:S-flow}\\
\dot q
&=
-8\kappa S q(1-q^2).
\label{eq:q-flow}
\end{align}
\end{subequations}
Thus every positive initial pair has \(|q|<1\), the imbalance flows toward
zero, and
\begin{equation}
6\kappa
\le
\frac{\dd}{\dd\ell}\frac1S
\le
14\kappa.
\label{eq:S-recip-bounds}
\end{equation}
Consequently
\begin{equation}
q(\ell)\longrightarrow0,
\qquad
S(\ell)\asymp\frac1\ell,
\label{eq:S-q-asymptotic}
\end{equation}
where \(f\asymp g\) means that \(f/g\) remains between two positive
constants for sufficiently large \(\ell\).  Consequently, both channels
accumulate infinite activity:
\begin{equation}
\int^\infty a^2(\ell)\,\dd\ell
=
\int^\infty b^2(\ell)\,\dd\ell
=
\infty.
\label{eq:common-cone-divergent-activity}
\end{equation}
The exact invariant relating \(q\) and \(S\), and the resulting faster
decay of the coupling imbalance, are given in
appendix~\ref{app:persistence-topology}.

It remains to show that the preceding behavior is stable to a finite
propagation-cone mismatch.  The logarithmic one-loop Yukawa-vertex and
quartic coefficients are real analytic functions of the kinetic matrices on
every compact SPD domain: after Feynman parametrization the combined
quadratic forms remain uniformly positive definite, and the ultraviolet
logarithmic coefficients are built from their inverses and determinants.
The Pauli internal algebra is unchanged by these spacetime coefficients.
On a compact domain, the mean-value theorem therefore converts analyticity
into uniform Lipschitz bounds proportional to the Thompson mismatch.
Appendix~\ref{appsub:quartic-trapping} gives the explicit estimate and shows
that initial normalized quartics in \(\mathcal U_*\) remain in a compact
perturbative neighborhood.  Hence, in a sufficiently small Thompson
neighborhood of the common-cone manifold, the two-channel Yukawa flow can be
written as
\begin{subequations}
\label{eq:perturbed-Sq-flow}
\begin{align}
\dot S
&=
-\kappa S^2
\left[
6+8q^2+\varepsilon_S(G,q)
\right],
\label{eq:perturbed-S}\\
\dot q
&=
-\kappa S
\left[
8q(1-q^2)+\varepsilon_q(G,q)
\right],
\label{eq:perturbed-q}
\end{align}
\end{subequations}
where, uniformly for
\(\DT(G)\le\delta\),
\begin{equation}
|\varepsilon_S(G,q)|
\le C_S\delta,
\qquad
|\varepsilon_q(G,q)|
\le C_q\delta.
\label{eq:epsilon-bounds}
\end{equation}
The analyticity statement, the Lipschitz estimate, and the constants needed
below are developed in appendix~\ref{app:persistence-topology}.

Fix \(q_0\in(0,1)\) and choose the trapping neighborhood
\(\mathcal U_*\) small enough that
\(\mathcal U_*\subset\{|q|<q_0\}\).  At zero cone mismatch,
eq.~\eqref{eq:q-flow} points strictly inward on both boundaries
\(q=\pm q_0\).  By
eq.~\eqref{eq:epsilon-bounds}, there is therefore a
\(\delta_0>0\) such that the finite-mismatch flow remains inward on
\(|q|=q_0\) whenever
\begin{equation}
\DT(G)\le\delta_0.
\label{eq:delta0-neighborhood}
\end{equation}
The bootstrap proceeds in a definite order.  The multiplicative coupling
equations first preserve \(a,b>0\).  With positive edge rates,
Theorem~\ref{thm:global-nonexpansion} then gives
\begin{equation}
\DT(\ell)\le\DT(0),
\label{eq:bootstrap-nonexpansion}
\end{equation}
so an initial condition with \(\DT(0)<\delta_0\) never leaves the
coefficient-control neighborhood.  The inward \(q\)-boundary therefore
implies
\begin{equation}
|q(\ell)|\le q_0,
\qquad
a^2,b^2
\ge
\frac{1-q_0}{2}\,S.
\label{eq:q-invariant-strip}
\end{equation}
Taking \(\delta_0\) smaller if needed, the first equation in
eq.~\eqref{eq:perturbed-Sq-flow} obeys uniform bounds
\begin{equation}
-c_2S^2
\le
\dot S
\le
-c_1S^2,
\qquad
0<c_1<c_2,
\label{eq:S-finite-mismatch-bounds}
\end{equation}
and hence
\begin{equation}
S(\ell)\asymp\frac1\ell.
\label{eq:S-finite-mismatch-asymptotic}
\end{equation}
Equations~\eqref{eq:q-invariant-strip} and
\eqref{eq:S-finite-mismatch-asymptotic} show that each Yukawa edge still
accumulates infinite activity.

Finally, the cone-flow rates on edge \(r\) equal \(g_r^2\) times positive
wave-function and determinant factors that are continuous on the same
compact SPD domain.  Choosing
\begin{equation}
\chi(\ell)=S(\ell)
\label{eq:chi-S}
\end{equation}
therefore gives constants \(m,M>0\) such that
\begin{equation}
0<m
\le
\frac{\alpha_r}{\chi},
\frac{\beta_r}{\chi}
\le M,
\label{eq:persistence-realized}
\end{equation}
while
\begin{equation}
s(\ell)
=
\int_0^\ell S(u)\,\dd u
\longrightarrow\infty.
\label{eq:clock-diverges}
\end{equation}
This realizes the hypotheses of
Theorem~\ref{thm:finite-window} dynamically within the field theory.

\subsection{An open Yukawa basin with common-cone locking}
\label{subsec:open-locking-basin}

The persistence bootstrap and the finite-window theorem together give the
following field-theoretic consequence.

\begin{corollary}[Finite-mismatch locking]
\label{cor:open-locking-basin}
Consider the two-channel Pauli--Dirac Yukawa realization
eq.~\eqref{eq:pauli-two-channel}, completed by the critical scalar potential
eq.~\eqref{eq:two-channel-quartic-potential}.  There exist
\(\delta_0>0\) and an open neighborhood \(\mathcal U_*\) of the normalized
fixed point \(p_*\) such that every initial condition satisfying
\begin{equation}
\begin{aligned}
a(0)&>0,
& b(0)&>0,\\
\DT(G(0))&<\delta_0,
& (q(0),u(0),v(0),w(0))&\in\mathcal U_*.
\end{aligned}
\label{eq:open-basin-initial}
\end{equation}
remains in a weak-coupling persistence region in which the Yukawa and
quartic couplings stay perturbative, and the accumulated interaction time
diverges.  Consequently,
\begin{equation}
\DT(\ell)\longrightarrow0
\qquad
(\ell\to\infty)
\label{eq:open-basin-locking}
\end{equation}
throughout this open set of genuinely finite-mismatch propagation
geometries.
\end{corollary}

The corollary connects Theorem~\ref{thm:finite-window} directly to the field
theory: the explicit Yukawa--quartic completion dynamically generates both a
compact weak-coupling basin and the persistent edge activity required for
locking on an open neighborhood of mismatched cones.  The next section turns from the
existence of global locking to a sharper question: how the interaction
topology controls the first nonzero short-time order at which an extremal
diameter begins to decrease.

\section{Interaction topology and the onset of locking}
\label{sec:topology-onset}

The finite-window theorem determines whether a persistently active connected
network locks.  We now ask a sharper short-time question: at what accumulated
interaction-time order does the initially maximal cone mismatch first begin
to decrease?  For a controlled two-plateau extremal class, the answer is set
exactly by graph distance.  The underlying distance-to-order mechanism has a
linear analogue in graph diffusion, where the short-time heat kernel between
two vertices begins at the power fixed by their combinatorial distance
\cite{KellerLenzMunchSchmidtTelcs2016}.  Here this distance-to-order mechanism is realized within a nonlinear matrix
flow that also contains generalized-eigenvalue branches, equality edges, and
the Thompson-diameter envelope.

\subsection{Two-plateau extremal data}
\label{subsec:two-plateau-data}

Fix a common simple active generalized-eigenvector direction \(x\) at
\(s=0\).  Normalize the two extremal quadratic-form values to
\begin{equation}
1<\Lambda,
\label{eq:plateau-values}
\end{equation}
and suppose every cone lies on one of the two \(x\)-plateaus,
\begin{equation}
x^TG_v(0)x
\in
\{1,\Lambda\}.
\label{eq:two-plateau}
\end{equation}
Let
\begin{equation}
\begin{aligned}
L_x&=\{v:\;x^TG_v(0)x=1\},\\
U_x&=\{v:\;x^TG_v(0)x=\Lambda\}.
\end{aligned}
\label{eq:plateau-sets}
\end{equation}
be the lower and upper node sets.  In addition, every graph edge internal to
one plateau is assumed to saturate the stronger equality condition of
section~\ref{subsec:kernel-rigidity} in the appropriate endpoint frame.  Thus
such an edge carries no first-order inward flux along the active direction.

Let
\(\mathcal A_x\subseteq L_x\times U_x\) denote the ordered endpoint pairs
whose simple \(x\)-branch realizes the initial Thompson diameter.  Let \(d(v,w)\) denote the shortest-path graph distance between nodes
\(v\) and \(w\).  For a node \(v\) and a node set \(S\), define
\begin{equation}
d(v,S)=\min_{w\in S}d(v,w)
\label{eq:distance-to-set}
\end{equation}
as the graph distance from \(v\) to the set.  Define, for every active pair,
\begin{equation}
q_{ij}
=
\min
\left\{
d(i,U_x),
d(j,L_x)
\right\},
\qquad
(i,j)\in\mathcal A_x,
\label{eq:qij}
\end{equation}
and
\begin{equation}
q
=
\max_{(i,j)\in\mathcal A_x}
q_{ij}.
\label{eq:qglobal-sec5}
\end{equation}
The quantity \(q_{ij}\) measures how deeply the two endpoints of the active
pair are shielded from the first opposite-plateau edge; \(q\) selects the
most deeply shielded branch among those supporting the initial diameter.

For later use, let
\(\lambda_{ij}^{(x)}(s)\) denote the analytic continuation of the simple
generalized-eigenvalue branch associated with \(x\) and the pair \((i,j)\),
with
\begin{equation}
\lambda_{ij}^{(x)}(0)=\Lambda,
\end{equation}
and define
\begin{equation}
D_{ij}(s)
=
\log \lambda_{ij}^{(x)}(s).
\label{eq:Dij-branch}
\end{equation}

\subsection{Propagation of the first nonzero coefficient}
\label{subsec:topology-propagation}

The endpoint-propagation mechanism of
section~\ref{subsec:endpoint-propagation} becomes quantitative for the data
above.  Orient a shortest path from a lower endpoint toward \(U_x\), or from
an upper endpoint toward \(L_x\).  The first edge crossing the plateau
boundary is genuinely mismatched and therefore generates a strictly inward
order-one coefficient along the propagated active direction.

Suppose inductively that the first nonzero inward coefficient at a node
lying \(m\) internal plateau edges away from such a boundary occurs at order
\(m+1\).  On a bosonic equality edge, differentiation of
\begin{equation}
\widehat\alpha_e(A-B)
\end{equation}
transfers that coefficient one order later with positive multiplier
\(\widehat\alpha_e(0)\).  On a fermionic equality edge,
eq.~\eqref{eq:Frechet-transfer} transfers it one order later with multiplier
\begin{equation}
2\widehat\beta_e(0)c(R)>0.
\label{eq:fermion-positive-transfer}
\end{equation}
Derivatives of the normalized rates and of the relative frame multiply
lower-order edge fluxes, which vanish at the induction step.  Hence they do
not change the first nonzero order or its sign.  If several shortest paths
reach the same endpoint, their leading coefficients add with the same inward
sign.

It follows that the branch of an active pair \((i,j)\) is stationary through
order \(q_{ij}-1\) and moves inward at order \(q_{ij}\):
\begin{equation}
D_{ij}(0)-D_{ij}(s)
=
C_{ij}s^{q_{ij}}
+
O(s^{q_{ij}+1}),
\qquad
C_{ij}>0.
\label{eq:pair-onset-sec5}
\end{equation}
The full Taylor-order induction, including multiple shortest paths and
moving-frame terms, is given in
appendix~\ref{app:persistence-topology}.

\subsection{Topology-sensitive onset}
\label{subsec:topology-theorem}

The pairwise expansion is then lifted to the Thompson diameter, which is a
maximum over branches.  Branches that do not
realize the initial diameter are separated from it at \(s=0\) and therefore
cannot support the envelope for sufficiently small \(s>0\).  Among the
initially active branches, those with smaller \(q_{ij}\) move inward at a
lower order and leave the maximum first.  The branches with
\(q_{ij}=q\) remain on the envelope longest and therefore set the first
nonzero order of the network diameter.

\begin{theorem}[Topology onset]
\label{thm:topology-onset}
Consider the normalized analytic cone flow
eq.~\eqref{eq:normalized-network-flow} on a finite connected graph, and
initial data satisfying the two-plateau conditions
eqs.~\eqref{eq:two-plateau}--\eqref{eq:plateau-sets} with a common simple
active direction \(x\) and equality on every edge internal to a plateau.
Let \(\mathcal A_x\), \(q_{ij}\), and \(q\) be defined by
eqs.~\eqref{eq:qij}--\eqref{eq:qglobal-sec5}.  Then each initially active
pair obeys
\begin{equation}
D_{ij}(0)-D_{ij}(s)
=
C_{ij}s^{q_{ij}}
+
O(s^{q_{ij}+1}),
\qquad
C_{ij}>0,
\label{eq:thm3-pairwise}
\end{equation}
and the Thompson diameter satisfies
\begin{equation}
\DT(0)-\DT(s)
=
C_qs^q
+
O(s^{q+1}),
\qquad
C_q>0.
\label{eq:thm3-global}
\end{equation}
\end{theorem}

\begin{proof}
Equation~\eqref{eq:thm3-pairwise} is the shortest-path induction summarized
in section~\ref{subsec:topology-propagation}.  For sufficiently small \(s\),
only branches in \(\mathcal A_x\) can support the initial diameter envelope.
Every branch with \(q_{ij}<q\) has already decreased at an earlier order.
The remaining branches with \(q_{ij}=q\) all have positive inward
coefficients.  Taking the smallest of these finitely many positive
order-\(q\) coefficients gives \(C_q>0\) and
eq.~\eqref{eq:thm3-global}.
\end{proof}

The theorem is an exact statement for the controlled two-plateau extremal
class.  Its exponent depends not on the graph diameter alone but on the
shielding depth of the branches that actually support the initial Thompson
diameter, through eq.~\eqref{eq:qglobal-sec5}.

\subsection{Four-node analytic coefficient}
\label{subsec:P4-coefficient}

A four-node isotropic path gives the coefficient directly, independently of
the general induction.  Consider
\begin{equation}
F_1-B_1-F_2-B_2
\label{eq:P4-path}
\end{equation}
with constant normalized fermionic and bosonic rates
\(\bar\beta,\bar\alpha>0\), respectively, and initial scalar cone values
\begin{equation}
(1,1,L,L),
\qquad
L>1.
\label{eq:P4-step}
\end{equation}
The isotropic submanifold
\begin{equation}
G_v(s)=g_v(s)\Id
\end{equation}
is invariant.  Writing
\begin{equation}
\cH(r\Id)=h(r)\Id,
\qquad
h(r)
=
\frac{4(\sqrt r-1)}
{3\sqrt r(\sqrt r+1)^2},
\label{eq:isotropic-h-sec5}
\end{equation}
one obtains
\begin{equation}
\DT(0)-\DT(s)
=
\bar\alpha\bar\beta
\left[
\frac{L-1}{6}
-
h(L^{-1})
\right]s^2
+
O(s^3).
\label{eq:P4-coefficient}
\end{equation}
Because \(L>1\) implies \(h(L^{-1})<0\), the coefficient is strictly
positive.  Equation~\eqref{eq:P4-coefficient} realizes \(q=2\) and fixes the
sign and normalization of the alternating bosonic and fermionic transfer
steps.

\subsection{Numerical onset exponents}
\label{subsec:onset-numerics}

We illustrate the exponent on the invariant isotropic submanifold.
The initial data are symmetric two-level steps on the paths
\(P_4,P_6,P_8\), for which
Theorem~\ref{thm:topology-onset} predicts
\begin{equation}
q=2,\qquad q=3,\qquad q=4,
\label{eq:path-q-predictions}
\end{equation}
respectively.  Constant normalized rates are used, so accumulated
interaction time equals the integration variable, \(s=\ell\).

Figure~\ref{fig:topology-onset}(a) plots
\(\DT(0)-\DT(s)\) together with \(s^q\) guides.  High-precision integration
with \(60\) decimal digits and a fixed fourth-order Runge--Kutta step
\(10^{-5}\) gives fitted powers
\begin{equation}
1.992389,\qquad
2.999501,\qquad
3.994740
\label{eq:onset-fits}
\end{equation}
over
\begin{equation}
3\times10^{-4}
\le s\le
3\times10^{-3}.
\label{eq:onset-fit-window}
\end{equation}
Panel~(b) shows the local logarithmic slope
\begin{equation}
q_{\rm eff}(s)
=
\frac{\dd\log[\DT(0)-\DT(s)]}
{\dd\log s},
\label{eq:qeff}
\end{equation}
which approaches the predicted integers as \(s\to0\).  Implementation
details and the associated numerical data are given in
appendix~\ref{appsub:onset-numerics}.

\begin{figure*}[t]
\centering
\includegraphics[width=0.96\textwidth]{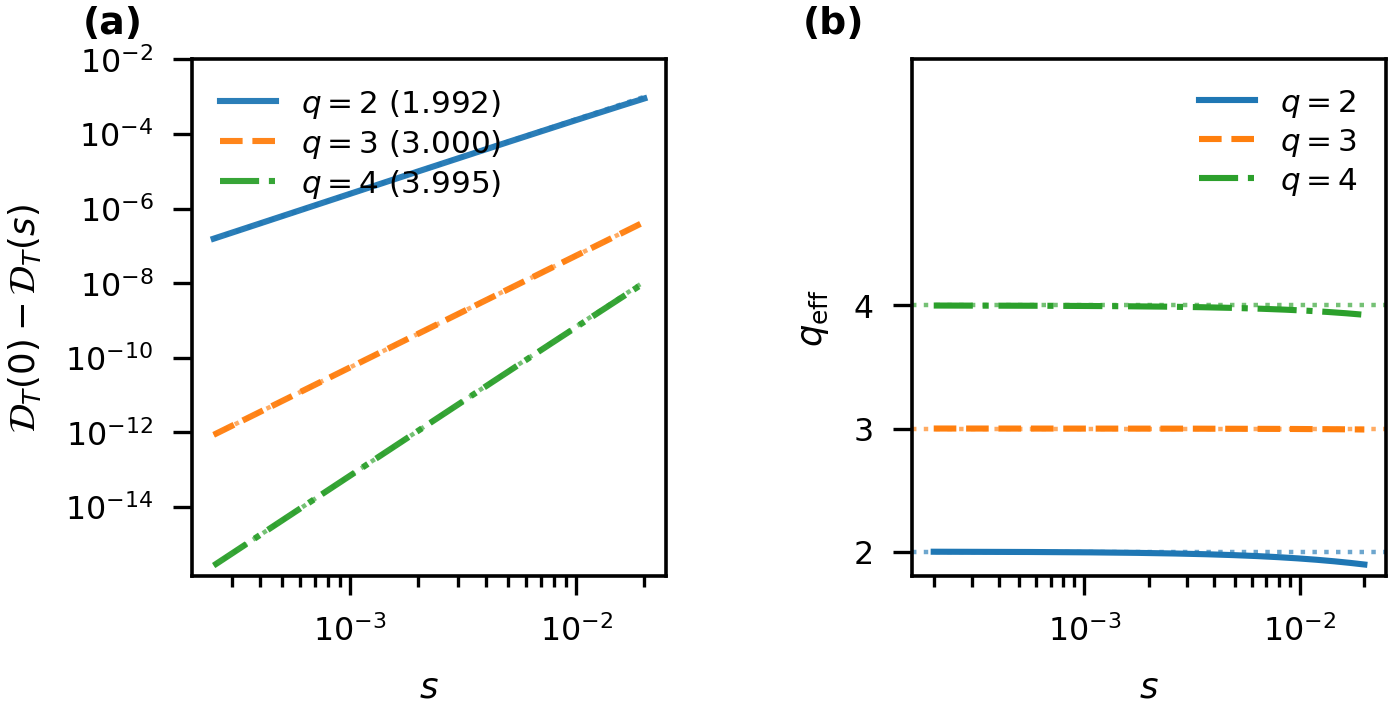}
\caption{
Topology-sensitive onset for symmetric two-plateau data on path graphs.
(a) Diameter decrease \(\mathcal D_T(0)-\mathcal D_T(s)\) for
\(P_4,P_6,P_8\), together with the predicted \(s^q\) scaling; the
parenthetical values in the legend are fitted powers over
eq.~\eqref{eq:onset-fit-window}.
(b) Effective exponent \(q_{\rm eff}(s)\), approaching
\(2,3,4\) toward short accumulated interaction time.
}
\label{fig:topology-onset}
\end{figure*}

The numerical calculation illustrates the controlled isotropic realization
of Theorem~\ref{thm:topology-onset}.  The theorem also covers finite
matrix-valued plateau data satisfying its active-direction and equality-edge
conditions.  Together with Theorem~\ref{thm:finite-window}, the result
separates two roles of topology: connectivity determines asymptotic relative
locking, while graph distance resolves the short-time order at which
strictness reaches shielded extremal sectors.

\section{Discussion and outlook}
\label{sec:discussion}

The analysis establishes a mechanism by which a connected Yukawa network
synchronizes a finite family of anisotropic propagation geometries.  It starts
from a field-theoretic one-loop flow that retains the full relative
positive-definite kinetic matrices at finite mismatch.  This matrix flow has
a global order structure:
its Thompson diameter is nonincreasing for arbitrary finite SPD mismatch,
and persistent activity on a connected graph upgrades this to a uniform
finite-window contraction,
\begin{equation}
\DT(s+T)\le\rho_T\DT(s),
\qquad
0<\rho_T<1.
\end{equation}
Accordingly,
\begin{equation}
\DT\longrightarrow0
\end{equation}
whenever the accumulated interaction time diverges.  In this sense the
infrared dynamics drives the interacting sectors to a common relative
propagation geometry.

The explicit persistence realization exhibits two distinct infrared regimes.
In strictly four dimensions the Yukawa couplings are marginally irrelevant.
Equations~\eqref{eq:S-finite-mismatch-asymptotic} and
\eqref{eq:clock-diverges} give
\begin{equation}
S(\ell)\asymp\ell^{-1},
\qquad
s(\ell)=\int_0^\ell S(u)\,\dd u\longrightarrow\infty,
\end{equation}
while the quartic couplings remain \(O(S)\).  Hence the interaction
couplings vanish in the infrared, but their accumulated action is still
infinite and locks the relative propagation matrices before the theory
reaches its Gaussian limit.  The four-dimensional endpoint is therefore a
free theory with one common spatial propagation cone, rather than a set of
free species retaining distinct cones.

The \(4-\varepsilon\) theory provides a complementary interacting
perspective.  On the equal NJLY ray,
eqs.~\eqref{eq:NJLY-matching} and
\eqref{eq:appA-NJLY-fixed-pair} give
\begin{equation}
g_*^2=\frac{\varepsilon}{12},
\qquad
S_*=2g_*^2=\frac{\varepsilon}{6}>0,
\end{equation}
so a persistent trajectory near that interacting critical point has an
interaction clock growing linearly, \(s(\ell)\sim S_*\ell\).  The contraction
mechanism therefore extends from the marginally irrelevant four-dimensional
regime to an interacting critical scaling regime with nonvanishing Yukawa
activity.  In four dimensions, the trapping theorem above supplies an explicit
open finite-mismatch basin.  An analogous trapping construction at
\(\varepsilon>0\) would establish the corresponding open interacting critical
basin.

The mechanism has both model-specific and structural ingredients.  The
model-specific input is carried by the Yukawa kernel
\(\cH(R)\): its inversion identity, extremal rank-one sign, and equality
rigidity are what make the global Thompson-diameter estimate possible.
The field-theoretic persistence realization is likewise specific to the
Pauli--Dirac subclass analyzed in
section~\ref{subsec:qft-persistence}.  By contrast, the route from
endpoint rigidity to finite-window contraction is more structural.  Once a
positive-cone flow supplies a congruence-invariant diameter, a rigid equality
condition, and a positive transfer coefficient across saturated edges,
connectedness converts local inward motion into a network-wide
synchronization mechanism.  The present Yukawa theory therefore provides a
concrete realization of a broader contraction architecture, whose extension
to other matrix-valued propagation sectors can be tested by checking those
three ingredients.

The result also sharpens the relation between velocity matching and emergent
Lorentz symmetry.  As reviewed in section~\ref{sec:introduction}, existing
RG analyses encompass scalar limiting velocities, anisotropic Dirac
dispersions, tilt, birefringence, and more general Lorentz-violating kinetic
structures.  Here the synchronized object is the full spatial propagation
quadratic form across a finite Yukawa interaction network.  Relative
eigenvalues and relative eigenframes are driven toward a common
characteristic geometry.  With temporal coefficients normalized and
time--space tilt set to zero, the one-loop flow provides an explicit route
from species-dependent anisotropic propagation to synchronization of the full
spatial characteristic quadratic forms.  Incorporating time--space tilt and
general Lorentzian tensors is the natural extension to a broader causal
geometry.

The relative and common-mode sectors therefore separate naturally.
Theorems~\ref{thm:global-nonexpansion} and
\ref{thm:finite-window} determine the relative sector through
\begin{equation}
\DT\to0,
\end{equation}
so all cones become asymptotically identical relative to one another.
Common-mode convergence in a fixed normalization is governed by the
common-mode drift and the detailed balance of the running rates.  Thus the
locking theorem determines whether a connected network develops one common
propagation geometry, while the common-mode dynamics determines which
terminal geometry is reached.

The topology-sensitive onset theorem resolves a complementary question.
Connectivity controls the asymptotic locking statement, while the graph
distance from an extremal endpoint to the first opposite plateau controls the
short-time order at which strictness arrives for the controlled two-plateau
class.  The exponent
\begin{equation}
q
=
\max_{(i,j)\in\mathcal A_x}
\min\{d(i,U_x),d(j,L_x)\}
\end{equation}
therefore measures shielding of the active Thompson-diameter branches rather
than a generic graph diameter.  This distinction explains why two connected
networks can share the same asymptotic locking conclusion while exhibiting
different early-time contraction laws.

The field-theoretic setting is the one-loop weak-coupling sector with
positive-definite spatial kinetic matrices, fixed interaction topology, and
vanishing time--space tilt.  In this setting, the global nonexpansion theorem
covers arbitrary finite SPD mismatch, while the explicit two-channel
Yukawa--quartic completion establishes running-coupling persistence on an open
finite-mismatch neighborhood of the common-cone manifold.  The compact-domain
formulation accommodates genuinely finite geometric mismatch while keeping
the loop expansion and the persistence bootstrap perturbatively controlled.

The logarithmic one-loop kinetic coefficients agree across the regulator
deformations compared in appendix~\ref{app:selfenergies}, including
dimensional regularization and smooth star-shaped cutoffs.  In
Lorentz-violating Yukawa theories, beta functions can depend on the
parameterization of redundant couplings across renormalization schemes
\cite{KarkiAltschul2022}.  The one-loop cone flow is stated in the convention specified above.
Physical relative-cone observables provide a natural scheme-robust language
for higher-order extensions.

Several extensions follow directly from this structure.  First, determining
the common-mode evolution would turn the locking theorem into a terminal
selection theory for the infrared cone.  Second, gauge systems replace a
single SPD propagation matrix by more general Fresnel or polarization
structures, providing a natural test of whether the same extremal-contraction
architecture survives beyond Yukawa matter.  Third, restoring time--space
tilt would promote the present spatial cone problem to a genuinely
Lorentzian one.  Finally, higher-loop corrections provide a direct test of
the structural stability of the finite-window mechanism: the relevant
question is whether the extremal sign, equality rigidity, and positive
transfer coefficient persist under perturbative deformation of the one-loop
vector field.

Taken together, the results establish for the Yukawa class analyzed here
that a universal spatial propagation cone can emerge as a collective infrared
property of the interaction network.  Local Yukawa interactions erase finite
matrix-valued
differences in both propagation rates and eigenframes across a connected
network.  In four dimensions this synchronization persists as the
interactions that generate it become marginally irrelevant: the Yukawa
couplings flow to zero while their accumulated RG action diverges, leaving a
Gaussian infrared theory with a common spatial cone.  Connectivity determines
the asymptotic locking, while graph distance controls the earliest order at
which strictness reaches shielded extremal sectors in the controlled
two-plateau class.  In this sense, common-cone locking is a collective infrared
property of the interaction network.

\begin{acknowledgments}
The author used ChatGPT (OpenAI, GPT-5.6 Sol) to assist with literature organization, mathematical cross-checking, and manuscript language and presentation. The author directed the tool through task-specific prompts and independently reviewed and verified all scientific content, derivations, numerical results, and citations and takes full responsibility for the manuscript.
\end{acknowledgments}

\section*{Data Availability}
The Python code and numerical data used to generate
Figs.~\ref{fig:global-locking} and~\ref{fig:topology-onset} are publicly
available as ancillary files to arXiv:2608.15661v2 \cite{ZengReproducibility2026}.

\appendix

\section{Internal Yukawa algebra and coupling RG}
\label{app:internal-algebra}

This appendix fixes the fermion conventions used in the two-channel
persistence analysis and derives the coefficients in
eqs.~\eqref{eq:two-channel-beta-epsilon} directly from the general one-loop
matrix Yukawa beta function.  The derivation also makes the one-loop
Pauli-channel closure of section~\ref{subsec:channel-closure} explicit.

\subsection{Dirac-to-Weyl map and coupling normalization}
\label{appsub:dirac-weyl}

We use the four-dimensional Weyl convention of
Ref.~\cite{PannellStergiou2023}.  For real scalars \(\phi_r\) and
two-component Weyl fermions \(\psi_A\), where capital indices such as
\(A,B\) label Weyl species and are summed when repeated, their Yukawa
interaction is
\begin{equation}
\mathcal L_Y
=
\frac12
\left(
Y_{r,AB}\phi_r\psi_A\psi_B+\text{h.c.}
\right),
\qquad
Y_r^T=Y_r .
\label{eq:appA-weyl-lagrangian}
\end{equation}
The dimensionless Yukawa matrices in that convention are rescaled by the
common factor \(4\pi\).

A Dirac multiplet with \(d\)-dimensional internal flavor space is written in
terms of \(2d\) Weyl fields.  For a purely scalar Dirac Yukawa matrix
\(y_r^D\) and vanishing pseudoscalar matrix, the block map of
Ref.~\cite{PannellStergiou2023} gives
\begin{equation}
Y_r
=
\frac12
\begin{pmatrix}
0&(y_r^D)^*\\
y_r^D&0
\end{pmatrix}.
\label{eq:appA-dirac-weyl-map}
\end{equation}
For the Hermitian involution channels used in this paper we set
\begin{equation}
y_r^D=2g_rT_r,
\qquad
T_r^\dagger=T_r,
\qquad
T_r^2=\Id_d,
\label{eq:appA-dirac-rescale}
\end{equation}
where \(\Id_d\) is the identity on the \(d\)-dimensional internal flavor
space.  Thus

\begin{equation}
Y_r=g_rK_r,
\qquad
K_r
=
\begin{pmatrix}
0&T_r^*\\
T_r&0
\end{pmatrix}.
\label{eq:appA-Klift}
\end{equation}
Accordingly, \(g_r\) is the Weyl-block amplitude, equivalently a common
rescaling of
the dimensionless Dirac Yukawa matrix.  This is the normalization denoted by
\(a,b\) in section~\ref{subsec:qft-persistence}; any overall conversion
factor is absorbed into the positive constant \(\kappa\) when the flow is
rewritten in the infrared-oriented variable.

Hermiticity of \(T_r\) implies
\begin{equation}
K_r^T=K_r,
\qquad
K_rK_r^\dagger=\Id_{2d}.
\label{eq:appA-Kproperties}
\end{equation}
If two internal channels satisfy
\begin{equation}
T_sT_rT_s=\eta_{sr}T_r,
\qquad
\eta_{sr}\in\{+1,-1\},
\label{eq:appA-Tconjugation}
\end{equation}
then the lifted Weyl matrices obey
\begin{equation}
K_sK_r^\dagger K_s
=
\eta_{sr}K_r.
\label{eq:appA-Kconjugation}
\end{equation}
For Pauli strings, \(\eta_{sr}=+1\) when the strings commute and
\(\eta_{sr}=-1\) when they anticommute.

\subsection{Finite-degree Pauli realization and one-loop closure}
\label{appsub:pauli-closure}

Let \(d=2^n\) and choose distinct nonidentity \(n\)-qubit Pauli strings for
the channels incident on a fermion node.  They satisfy
\begin{equation}
\Tr(T_rT_s)=d\,\delta_{rs},
\label{eq:appA-pauli-orthogonality}
\end{equation}
where \(\delta_{rs}\) is the Kronecker delta.  Every pair either commutes or
anticommutes.  Since the number of
nonidentity strings is \(4^n-1\), any finite node degree can be realized by
taking \(n\) sufficiently large.

For four-dimensional Weyl fermions, the scalar anomalous-dimension tensor
entering the beta function is
\begin{equation}
Z_{rs}
=
\Tr\!\left(
Y_rY_s^\dagger+Y_r^\dagger Y_s
\right).
\label{eq:appA-Zrs}
\end{equation}
Equations~\eqref{eq:appA-Klift} and
\eqref{eq:appA-pauli-orthogonality} give
\begin{equation}
Z_{rs}
=
4d\,g_r^2\,\delta_{rs}.
\label{eq:appA-Zdiagonal}
\end{equation}
Hence off-diagonal scalar wave-function mixing vanishes in this channel
basis.

Let \(\varepsilon\) denote the deviation from four spacetime dimensions and
let \(\mu\) be the renormalization scale.  With
\(\beta_{Y_r}=\mu\,\dd Y_r/\dd\mu\), the general one-loop Weyl Yukawa beta
function can be written as
\begin{align}
\beta_{Y_r}
={}&
-\frac{\varepsilon}{2}Y_r
+
\frac12\sum_s
\Bigl(
Y_sY_s^\dagger Y_r
+Y_rY_s^\dagger Y_s
\nonumber\\
&\hspace{35mm}
+4Y_sY_r^\dagger Y_s
+Z_{rs}Y_s
\Bigr).
\label{eq:appA-general-weyl-beta}
\end{align}
Using
\begin{equation}
Y_sY_s^\dagger=g_s^2\Id_{2d},
\qquad
Y_sY_r^\dagger Y_s
=
\eta_{sr}g_s^2Y_r,
\label{eq:appA-beta-algebra}
\end{equation}
and eq.~\eqref{eq:appA-Zdiagonal}, every term in
eq.~\eqref{eq:appA-general-weyl-beta} is proportional to the original
external channel \(K_r\).  The prescribed Pauli-string Yukawa subspace is
therefore invariant under the one-loop coupling flow.

Projecting onto \(K_r\) gives the general channel-amplitude equation
\begin{equation}
\beta_{g_r}
=
-\frac{\varepsilon}{2}g_r
+
g_r
\left[
\sum_s(1+2\eta_{sr})g_s^2
+
2d\,g_r^2
\right],
\label{eq:appA-general-channel-beta}
\end{equation}
where \(\eta_{rr}=1\).  Equivalently, the self-channel coefficient is
\(2d+3\), while a distinct commuting or anticommuting channel contributes
\(3\) or \(-1\), respectively.

\subsection{\texorpdfstring{Explicit two-channel coefficients \(7:-1\)}{Explicit two-channel coefficients 7:-1}}
\label{appsub:seven-minus-one}

For the persistence realization of section~\ref{subsec:qft-persistence},
take
\begin{equation}
d=2,
\qquad
T_1=\sigma_x,
\qquad
T_2=\sigma_z.
\label{eq:appA-explicit-pauli}
\end{equation}
The two matrices anticommute, so
\begin{equation}
\eta_{11}=\eta_{22}=1,
\qquad
\eta_{12}=\eta_{21}=-1.
\label{eq:appA-etas}
\end{equation}
Writing \(g_1=a\) and \(g_2=b\),
eq.~\eqref{eq:appA-general-channel-beta} gives
\begin{subequations}
\label{eq:appA-twochannel-beta}
\begin{align}
\beta_a
&=
-\frac{\varepsilon}{2}a
+
a(7a^2-b^2),
\label{eq:appA-beta-a}\\
\beta_b
&=
-\frac{\varepsilon}{2}b
+
b(7b^2-a^2).
\label{eq:appA-beta-b}
\end{align}
\end{subequations}

It is useful to display the coefficient bookkeeping explicitly.  For
\(\beta_a\), the two external-leg terms in
eq.~\eqref{eq:appA-general-weyl-beta} contribute
\begin{equation}
+a(a^2+b^2),
\end{equation}
the conjugation term contributes
\begin{equation}
+2a(a^2-b^2),
\end{equation}
and the diagonal scalar wave-function term contributes
\begin{equation}
+4a^3.
\end{equation}
Thus
\begin{equation}
7=1+2+4,
\qquad
-1=1-2+0.
\label{eq:appA-seven-minus-one-breakdown}
\end{equation}
The equation for \(b\) follows by exchange symmetry.  In particular,
anticommuting Pauli conjugation in the vertex term gives the negative cross
coefficient directly.

\subsection{Equal ray and NJLY normalization}
\label{appsub:njly-normalization}

On the symmetric ray
\begin{equation}
a=b=g,
\end{equation}
eqs.~\eqref{eq:appA-twochannel-beta} reduce to
\begin{equation}
\beta_g
=
-\frac{\varepsilon}{2}g+6g^3,
\label{eq:appA-equal-beta}
\end{equation}
with nontrivial one-loop fixed point
\begin{equation}
g_*^2=\frac{\varepsilon}{12}.
\label{eq:appA-gstar}
\end{equation}
This normalization provides an independent check against the standard
Nambu--Jona-Lasinio--Yukawa result.  In the notation of
Ref.~\cite{PannellStergiou2023}, let \(N_D\) denote the number of Dirac
flavors and define \(N=4N_D\).  Restoring the overall \(\varepsilon\)
scaling suppressed in the fixed-point coordinates of that reference,
eq.~(3.10) gives the NJLY fixed point
\begin{equation}
y_*^2=\frac{\varepsilon}{N+4},
\qquad
N=4N_D.
\label{eq:appA-njly-source}
\end{equation}
For \(N_D=2\), \(N=8\), and hence
\(y_*^2=\varepsilon/12\), in agreement with
eq.~\eqref{eq:appA-gstar}.  The corresponding one-real-scalar GNY normalization provides a useful
contrast.  With the same restoration of \(\varepsilon\),
eq.~(3.2) of Ref.~\cite{PannellStergiou2023} gives
\(y_*^2=\varepsilon/(N+6)\) for GNY,
and the same denominator follows directly from
eq.~\eqref{eq:appA-general-channel-beta} after suppressing all distinct
Yukawa channels:
\begin{equation}
\beta_g
=
-\frac{\varepsilon}{2}g+(2d+3)g^3,
\qquad
 g_*^2=\frac{\varepsilon}{4d+6}
      =\frac{\varepsilon}{N+6}.
\label{eq:appA-gny-check}
\end{equation}
On the two-channel anticommuting equal ray, the second channel contributes
\(-g^3\), changing the coefficient from \(2d+3\) to \(2d+2\), and hence the
denominator from \(N+6\) to \(N+4\).  The NJLY agreement therefore serves
as an independent normalization check on the coefficients \(7\) and \(-1\)
derived in
eq.~\eqref{eq:appA-seven-minus-one-breakdown}.

\subsection{Infrared-oriented four-dimensional flow}
\label{appsub:IR-flow}

At \(\varepsilon=0\), the beta function
eq.~\eqref{eq:appA-twochannel-beta} is marginally irrelevant toward the
infrared.  Reversing from the conventional ultraviolet-oriented
\(\log\mu\) variable to the infrared-oriented variable \(\ell\), and
allowing for the common positive normalization used in the cone equations,
gives
\begin{subequations}
\label{eq:appA-IR-twochannel}
\begin{align}
\dot a
&=
-\kappa a(7a^2-b^2),
\\
\dot b
&=
-\kappa b(7b^2-a^2).
\end{align}
\end{subequations}
Here \(\kappa>0\) is a common convention-dependent normalization.
This is the flow quoted in eq.~\eqref{eq:two-channel-IR}.

The multiplicative form
\begin{equation}
\dot a=aF(a,b),
\qquad
\dot b=bH(a,b)
\end{equation}
makes the positive quadrant invariant.  Thus
\(a(0),b(0)>0\) implies \(a(\ell),b(\ell)>0\) for every finite RG time on
which the perturbative solution exists.  This positivity property provides the first step in the finite-mismatch
persistence bootstrap of
section~\ref{subsec:qft-persistence}; it precedes the use of
Theorem~\ref{thm:global-nonexpansion}.

\subsection{Closed quartic completion and normalized infrared flow}
\label{appsub:quartic-completion}

The two discrete transformations stated in
section~\ref{subsec:qft-persistence} forbid odd mixed scalar monomials.
After tuning the relevant scalar masses to the critical surface, the
renormalizable scalar potential is therefore the three-coupling potential
in eq.~\eqref{eq:two-channel-quartic-potential}.  In the same rescaled
convention as eq.~\eqref{eq:appA-general-weyl-beta}, projection of the
general one-loop scalar--fermion quartic beta function
\cite{PannellStergiou2023} gives
\begin{subequations}
\label{eq:appA-quartic-betas}
\begin{align}
\beta_{\lambda_1}
={}&
-\varepsilon\lambda_1
+3(\lambda_1^2+\lambda_{12}^2)
+16a^2\lambda_1
-96a^4,
\label{eq:appA-beta-lambda1}\\
\beta_{\lambda_2}
={}&
-\varepsilon\lambda_2
+3(\lambda_2^2+\lambda_{12}^2)
+16b^2\lambda_2
-96b^4,
\label{eq:appA-beta-lambda2}\\
\beta_{\lambda_{12}}
={}&
-\varepsilon\lambda_{12}
+\lambda_{12}(\lambda_1+\lambda_2)
+4\lambda_{12}^2
\nonumber\\
&+8(a^2+b^2)\lambda_{12}
-32a^2b^2 .
\label{eq:appA-beta-lambda12}
\end{align}
\end{subequations}
Here
\(\beta_{\lambda_I}=\mu\,\dd\lambda_I/\dd\mu\), with
\(I\in\{1,2,12\}\).  The first quadratic terms are the scalar-loop
contributions, the terms proportional to
\(\lambda_I a^2\) or \(\lambda_I b^2\) are generated by scalar
wave-function renormalization, and the negative quartic Yukawa terms come
from the fermion box.  Because
\(T_1T_2+T_2T_1=0\), a constant scalar background obeys
\[
(a\phi_1T_1+b\phi_2T_2)^2
=
(a^2\phi_1^2+b^2\phi_2^2)\Id_2,
\]
where \(\Id_2\) is the identity on the two-dimensional internal flavor
space.  This makes the relative \(a^4:a^2b^2:b^4\) box coefficients
manifest.  The one-loop renormalizability of scalar--Yukawa theories with
Lorentz-violating kinetic terms is consistent with the general analysis of
Ref.~\cite{FerreroAltschul2011}.

On the equal Yukawa ray \(a=b=g\), the combined internal-flavor/scalar
rotation generated by \(\sigma_y\) leaves the interaction invariant.  The
quartic \(O(2)\) ray is
\begin{equation}
\lambda_1=\lambda_2=3\lambda,
\qquad
\lambda_{12}=\lambda.
\label{eq:appA-O2-quartic-ray}
\end{equation}
Substitution into eqs.~\eqref{eq:appA-quartic-betas} reduces all three
equations to
\begin{equation}
\beta_\lambda
=
-\varepsilon\lambda
+
10\lambda^2
+
16\lambda g^2
-
32g^4.
\label{eq:appA-NJLY-quartic-beta}
\end{equation}
Together with eq.~\eqref{eq:appA-equal-beta}, this gives
\begin{equation}
g_*^2=\frac{\varepsilon}{12},
\qquad
\lambda_*=\frac{2}{15}\varepsilon,
\qquad
\frac{\lambda_*}{g_*^2}=\frac85,
\label{eq:appA-NJLY-fixed-pair}
\end{equation}
which is the \(N_D=2\) NJLY fixed point in the same normalization.

For the four-dimensional infrared flow set \(\varepsilon=0\) and use the
infrared-oriented variable \(\ell\) of
section~\ref{subsec:qft-persistence}.  Define
\begin{equation}
\begin{aligned}
S&=a^2+b^2>0,
& q&=\frac{a^2-b^2}{S},\\
u&=\frac{\lambda_1}{S},
& v&=\frac{\lambda_2}{S},
& w&=\frac{\lambda_{12}}{S}.
\end{aligned}
\label{eq:appA-normalized-quartic-defs}
\end{equation}
Introduce a normalized quartic clock \(\vartheta\) by
\begin{equation}
\frac{\dd\vartheta}{\dd\ell}
=
\kappa S,
\label{eq:appA-vartheta}
\end{equation}
and let a prime denote \(\dd/\dd\vartheta\).  Combining
eqs.~\eqref{eq:appA-quartic-betas} with the Yukawa flow gives
\begin{subequations}
\label{eq:appA-normalized-quartic-flow}
\begin{align}
q'
={}&
-8q(1-q^2),
\label{eq:appA-q-prime}\\
u'
={}&
(-2-8q+8q^2)u
-3u^2-3w^2
+24(1+q)^2,
\label{eq:appA-u-prime}\\
v'
={}&
(-2+8q+8q^2)v
-3v^2-3w^2
+24(1-q)^2,
\label{eq:appA-v-prime}\\
w'
={}&
(-2+8q^2-u-v)w
-4w^2
+8(1-q^2).
\label{eq:appA-w-prime}
\end{align}
\end{subequations}
The point
\begin{equation}
p_*
=
(q_*,u_*,v_*,w_*)
=
\left(0,\frac{12}{5},\frac{12}{5},\frac45\right)
\label{eq:appA-pstar}
\end{equation}
is a fixed point.  Writing the deviation vector as
\begin{equation}
\xi
=
(q,\ u-u_*,\ v-v_*,\ w-w_*)^T,
\label{eq:appA-xi}
\end{equation}
the Jacobian \(J_*=\partial(q',u',v',w')/\partial(q,u,v,w)\) at \(p_*\)
is
\begin{equation}
J_*
=
\begin{pmatrix}
-8 & 0 & 0 & 0\\
144/5 & -82/5 & 0 & -24/5\\
-144/5 & 0 & -82/5 & -24/5\\
0 & -4/5 & -4/5 & -66/5
\end{pmatrix},
\label{eq:appA-quartic-Jacobian}
\end{equation}
with eigenvalues
\begin{equation}
-8,\qquad
-18,\qquad
-\frac{58}{5},\qquad
-\frac{82}{5}.
\label{eq:appA-quartic-eigenvalues}
\end{equation}
All eigenvalues have negative real part, so \(p_*\) is a hyperbolic
infrared sink of the common-cone normalized coupling flow.  In particular,
there is an open common-cone neighborhood in which the quartic ratios remain
bounded and converge toward \(p_*\), while the unnormalized quartics satisfy
\(\lambda_I=O(S)\) and therefore remain perturbative as \(S\to0\).

\section{One-loop finite-matrix self-energy derivations}
\label{app:selfenergies}

This appendix supplies the loop derivations underlying
section~\ref{sec:yukawa-flow}.  We keep the full relative positive-definite
spatial matrix throughout the fermion self-energy, derive the scalar
two-point counterterm with its complete quadratic form, and collect the
covariance, isotropic, and Fr\'echet identities used later in the locking
proof.
We use the convention \(\log(\Lambda_{\rm UV}/\mu)\), where
\(\Lambda_{\rm UV}\) is the ultraviolet regulator scale and \(\mu\) is the
renormalization scale.  Other standard logarithmic conventions amount to a
positive rescaling of the RG-time normalization.

\subsection{Fermion self-energy in the identity frame}
\label{appsub:fermion-selfenergy}

Consider one Yukawa edge with fermion propagation matrix \(A=e^Te\) and
scalar propagation matrix \(B\).  By polar decomposition,
\begin{equation}
e=O A^{1/2},
\qquad
O\in O(3),
\label{eq:appB-polar}
\end{equation}
where \(O(3)\) denotes the real orthogonal group.  The orthogonal factor can be absorbed into the spatial Dirac frame.
The loop momentum transformation
\begin{equation}
\boldsymbol K=A^{1/2}\boldsymbol k
\label{eq:appB-fermion-whitening}
\end{equation}
then places the fermion propagator in the identity frame.  The scalar
quadratic form becomes
\begin{equation}
R=A^{-1/2}BA^{-1/2},
\label{eq:appB-relative-R}
\end{equation}
while the positive Jacobian
\((\det A)^{-1/2}=|\det e|^{-1}\) is absorbed into the edge-rate
normalization when the result is returned to the original coordinates.

Let \(k=(k_0,\boldsymbol k)\) and \(q=(q_0,\boldsymbol q)\) denote
Euclidean four-momenta.  In the identity frame the fermion and scalar
propagators are, respectively,
\begin{equation}
S(k)
=
\frac{\gamma_0k_0+\boldsymbol\gamma\cdot\boldsymbol k}
{k_0^2+\boldsymbol k^2},
\qquad
D_R(q)
=
\frac{1}{q_0^2+\boldsymbol q^TR\boldsymbol q}.
\label{eq:appB-propagators}
\end{equation}
Since the internal Yukawa factor obeys \(T^2=\Id\), the logarithmically
divergent self-energy is proportional to the identity in internal flavor
space.  Denoting the Yukawa coupling on this edge by \(g\),
\begin{align}
\Sigma(p)
={}&
g^2
\int\frac{\dd k_0\,\dd^3k}{(2\pi)^4}
\left(\gamma_0k_0+\boldsymbol\gamma\cdot\boldsymbol k\right)
\nonumber\\
&\times
(k_0^2+\boldsymbol k^2)^{-1}
\nonumber\\
&\times
\left[
(p_0-k_0)^2+
(\boldsymbol p-\boldsymbol k)^T
R
(\boldsymbol p-\boldsymbol k)
\right]^{-1}.
\label{eq:appB-Sigma}
\end{align}

Using
\begin{equation}
\frac1{ab}
=
\int_0^1\frac{\dd y}{[(1-y)a+yb]^2},
\label{eq:appB-Feynman}
\end{equation}
define
\begin{equation}
M_y(R)=(1-y)\Id+yR.
\label{eq:appB-My}
\end{equation}
The combined denominator before completing the square is
\begin{align}
&(1-y)(k_0^2+\boldsymbol k^2)
+
y\left[
(p_0-k_0)^2+
(\boldsymbol p-\boldsymbol k)^TR(\boldsymbol p-\boldsymbol k)
\right]
\nonumber\\
&\quad=
(k_0-yp_0)^2
+
\left(
\boldsymbol k-yM_y^{-1}R\boldsymbol p
\right)^T
M_y
\nonumber\\
&\qquad\times
\left(
\boldsymbol k-yM_y^{-1}R\boldsymbol p
\right)
+
\Delta_y(p),
\label{eq:appB-complete-square}
\end{align}
where
\begin{equation}
\Delta_y(p)
=
y(1-y)
\left[
p_0^2+
\boldsymbol p^TRM_y^{-1}\boldsymbol p
\right].
\label{eq:appB-Delta}
\end{equation}
The shifted variables
\begin{equation}
q_0=k_0-yp_0,
\qquad
\boldsymbol q=
\boldsymbol k-yM_y^{-1}R\boldsymbol p
\label{eq:appB-shifts}
\end{equation}
remove all terms linear in the loop momentum.  Odd terms vanish upon
integration, leaving the external-momentum numerator
\begin{equation}
y\gamma_0p_0
+
y\gamma_i[M_y^{-1}R]_{ij}p_j.
\label{eq:appB-external-numerator}
\end{equation}
The coefficient of the ultraviolet logarithm multiplying the linear external
momentum is independent of \(\Delta_y\); \(\Delta_y\) enters the finite part.

\subsection{Logarithmic integral and regulator independence}
\label{appsub:log-integral}

For \(M\in\SPD(3)\) and \(\Delta>0\), define the scalar integral required
by the fermion self-energy as
\begin{equation}
I(M,\Delta)
=
\int\frac{\dd q_0\,\dd^3q}{(2\pi)^4}
\frac1{[q_0^2+\boldsymbol q^TM\boldsymbol q+\Delta]^2}.
\label{eq:appB-I}
\end{equation}
The spatial linear transformation
\begin{equation}
\boldsymbol Q=M^{1/2}\boldsymbol q
\end{equation}
gives
\begin{equation}
I(M,\Delta)
=
\frac1{\sqrt{\det M}}
\int\frac{\dd q_0\,\dd^3Q}{(2\pi)^4}
\frac1{(q_0^2+\boldsymbol Q^2+\Delta)^2}.
\label{eq:appB-I-transform}
\end{equation}
With a spherical four-momentum cutoff and the convention
\(\log(\Lambda_{\rm UV}/\mu)\),
\begin{equation}
I(M,\Delta)
=
\frac1{8\pi^2\sqrt{\det M}}
\log\frac{\Lambda_{\rm UV}}{\mu}
+
O(1).
\label{eq:appB-I-log}
\end{equation}

The same determinant factor follows from a proper-time cutoff.  Using
\begin{equation}
\frac1{X^2}
=
\int_0^\infty \dd t\,t\,e^{-tX}
\label{eq:appB-Schwinger}
\end{equation}
and the Gaussian integrals gives
\begin{equation}
I(M,\Delta)
=
\frac1{16\pi^2\sqrt{\det M}}
\int_0^\infty\frac{\dd t}{t}\,e^{-t\Delta}.
\label{eq:appB-propertime-I}
\end{equation}
Restricting the logarithmic proper-time interval to
\(t\in[\Lambda_{\rm UV}^{-2},\mu^{-2}]\) reproduces
eq.~\eqref{eq:appB-I-log}.  More generally, replacing the asymptotic
spherical boundary by any smooth bounded star-shaped cutoff changes the
radial upper limit by an angle-dependent finite factor.  The coefficient of
\(\int \dd r/r\), and hence the logarithmic matrix kernel, is unchanged;
only the finite part depends on the cutoff shape.

Substituting eq.~\eqref{eq:appB-I-log} into
eq.~\eqref{eq:appB-Sigma} gives
\begin{align}
\Sigma_{\log}(p)
={}&
\kappa_4 g^2
\log\frac{\Lambda_{\rm UV}}{\mu}
\left[
f_0(R)\gamma_0p_0
+
\gamma_iF_{ij}(R)p_j
\right],
\nonumber\\
&\kappa_4=\frac1{8\pi^2}.
\label{eq:appB-Sigma-log}
\end{align}
with
\begin{align}
f_0(R)
&=
\int_0^1
\frac{y\,\dd y}{\sqrt{\det M_y(R)}},
\label{eq:appB-f0}\\
F(R)
&=
\int_0^1
\frac{
yM_y(R)^{-1}R\,\dd y
}{
\sqrt{\det M_y(R)}
}.
\label{eq:appB-F}
\end{align}

\subsection{Temporal normalization and the finite-mismatch kernel}
\label{appsub:H-derivation}

For one infinitesimal infrared shell, write the positive logarithmic
coefficient as \(\delta z_\psi>0\).  In the identity fermion frame the
renormalized inverse kinetic operator has the form
\begin{equation}
(1+\delta z_\psi f_0)\gamma_0p_0
+
\gamma_i
\left[
\delta_{ij}+\delta z_\psi F_{ij}
\right]p_j.
\label{eq:appB-fermion-before-time-normalization}
\end{equation}
Restoring the temporal coefficient to unity changes the spatial frame by
\begin{equation}
e_{\rm new}
=
\Id+
\delta z_\psi
\left[
F(R)-f_0(R)\Id
\right]
+
O(\delta z_\psi^2).
\label{eq:appB-e-update}
\end{equation}
This identifies
\begin{equation}
\mathcal H(R)
=
F(R)-f_0(R)\Id.
\label{eq:appB-H-def}
\end{equation}
Using
\begin{equation}
yM_y^{-1}R
=
\Id-(1-y)M_y^{-1},
\label{eq:appB-My-identity}
\end{equation}
one obtains the finite-mismatch representation
\begin{equation}
\mathcal H(R)
=
\int_0^1\dd y\,
\frac{
(1-y)
\left[
\Id-M_y(R)^{-1}
\right]
}{
\sqrt{\det M_y(R)}
}.
\label{eq:appB-H}
\end{equation}
The derivation retains the full finite relative matrix \(R\) throughout.

Since \(R\) is symmetric positive definite, every \(M_y(R)\) is a polynomial
in \(R\) followed by inversion.  Therefore \(\mathcal H(R)\) is a spectral
matrix function:
\begin{equation}
\begin{aligned}
[\mathcal H(R),R]&=0,\\
\mathcal H(URU^T)&=U\mathcal H(R)U^T,
\qquad U\in O(3).
\end{aligned}
\label{eq:appB-H-spectral}
\end{equation}
The identity-frame quadratic form \(A=e^Te\) consequently changes as
\begin{equation}
\delta A
=
2\delta z_\psi\,\mathcal H(R)
+
O(\delta z_\psi^2).
\label{eq:appB-A-identity-update}
\end{equation}

\subsection{Scalar bubble and scalar-cone flow}
\label{appsub:scalar-bubble}

We next compute the momentum-dependent scalar counterterm.  Return to a
general fermion frame \(A=e^Te\), rotate away the orthogonal factor in
eq.~\eqref{eq:appB-polar}, and define
\begin{equation}
K_0=k_0,
\qquad
\boldsymbol K=A^{1/2}\boldsymbol k,
\qquad
P_0=p_0,
\qquad
\boldsymbol P=A^{1/2}\boldsymbol p.
\label{eq:appB-scalar-whitening}
\end{equation}
The measure contributes
\begin{equation}
\dd^3k
=
\frac{\dd^3K}{\sqrt{\det A}}
=
\frac{\dd^3K}{|\det e|}.
\label{eq:appB-scalar-Jacobian}
\end{equation}

Let
\begin{equation}
n_T=\Tr_{\mathcal V}(T^2).
\label{eq:appB-nT}
\end{equation}
For the involution channels used here,
\(n_T=\dim\mathcal V\).  Defining \(\Pi(p)\) as the contribution added to the
inverse scalar propagator, the one-loop fermion bubble is
\begin{equation}
\Pi(p)
=
-\frac{g^2n_T}{|\det e|}
\int\frac{\dd^4K}{(2\pi)^4}
\frac{
\operatorname{tr}_D[
\not\!K(\not\!K+\not\!P)
]
}{
K^2(K+P)^2}.
\label{eq:appB-Pi}
\end{equation}
Using
\begin{equation}
\begin{aligned}
\operatorname{tr}_D[\not\!K(\not\!K+\not\!P)]
&=4K\cdot(K+P)\\
&=2\left[K^2+(K+P)^2-P^2\right].
\end{aligned}
\label{eq:appB-Dirac-trace}
\end{equation}
the first two terms renormalize only the momentum-independent local scalar
operators in a cutoff treatment and vanish as scaleless integrals in
dimensional regularization.  The logarithmic kinetic contribution is
therefore
\begin{align}
\Pi_{\log}(p)
&=
\frac{2g^2n_T}{|\det e|}
P^2
\int\frac{\dd^4K}{(2\pi)^4}
\frac1{K^2(K+P)^2}
\nonumber\\
&=
\frac{g^2n_T}{4\pi^2|\det e|}
\left(
p_0^2+\boldsymbol p^TA\boldsymbol p
\right)
\log\frac{\Lambda_{\rm UV}}{\mu}
+
O(1).
\label{eq:appB-Pi-log}
\end{align}
The equality uses the same logarithmic convention as
eq.~\eqref{eq:appB-I-log}.  In dimensional regularization the identical
coefficient appears as the residue of the \(1/\varepsilon\) pole after the
standard conversion between \(2/\varepsilon\) and
\(\log(\Lambda_{\rm UV}^2/\mu^2)\).

Let
\begin{equation}
\delta z_\phi
=
\frac{g^2n_T}{4\pi^2|\det e|}\,\delta\ell
\label{eq:appB-dzphi}
\end{equation}
in the present shell convention.  Before temporal normalization the scalar
inverse kinetic operator is
\begin{equation}
(1+\delta z_\phi)p_0^2
+
\boldsymbol p^T
\left(
B+\delta z_\phi A
\right)
\boldsymbol p.
\label{eq:appB-scalar-before-normalization}
\end{equation}
Dividing by the temporal coefficient gives
\begin{equation}
B_{\rm new}
=
B+\delta z_\phi(A-B)
+
O(\delta z_\phi^2),
\label{eq:appB-Bupdate}
\end{equation}
and hence the edge flow
\begin{equation}
\dot B=\alpha(A-B),
\qquad
\alpha>0.
\label{eq:appB-Bflow}
\end{equation}
For the normalization used in
eq.~\eqref{eq:appB-dzphi},
\(\alpha=g^2n_T/(4\pi^2|\det e|)\) before further positive
wave-function/coupling reparametrizations.  The global locking analysis uses the positivity and regularity of the
resulting edge rate.

Scalar quartic tadpoles are momentum independent at one loop and therefore
do not modify eq.~\eqref{eq:appB-Bflow}.

\subsection{Return to a general frame and congruence covariance}
\label{appsub:Q-covariance}

The identity-frame update can be pulled back without choosing a preferred
eigenbasis.  Define
\begin{equation}
\mathcal Q(A,B)
=
A^{1/2}
\mathcal H(A^{-1/2}BA^{-1/2})
A^{1/2}.
\label{eq:appB-Q}
\end{equation}
Then the edge contribution to the fermion quadratic form is
\begin{equation}
\dot A=2\beta\,\mathcal Q(A,B),
\qquad
\beta>0,
\label{eq:appB-Aflow}
\end{equation}
where \(\beta\) contains the positive loop coefficient, the coupling squared,
and the Jacobian and wave-function factors appropriate to the chosen
coordinate convention.

The principal square root is not itself congruence covariant, whereas
\(\mathcal Q\) is.  For \(S\in GL(3,\mathbb R)\), set
\begin{equation}
U
=
(SAS^T)^{-1/2}SA^{1/2}.
\label{eq:appB-U}
\end{equation}
Then
\begin{equation}
UU^T=\Id,
\label{eq:appB-U-orthogonal}
\end{equation}
and
\begin{equation}
(SAS^T)^{-1/2}
SBS^T
(SAS^T)^{-1/2}
=
U
(A^{-1/2}BA^{-1/2})
U^T.
\label{eq:appB-relative-similarity}
\end{equation}
Using the orthogonal equivariance in
eq.~\eqref{eq:appB-H-spectral} and
\begin{equation}
(SAS^T)^{1/2}U=SA^{1/2},
\label{eq:appB-square-root-identity}
\end{equation}
we obtain
\begin{equation}
\mathcal Q(SAS^T,SBS^T)
=
S\mathcal Q(A,B)S^T.
\label{eq:appB-Q-covariance}
\end{equation}
This is the covariance property used throughout the main text.

\subsection{Isotropic closed forms}
\label{appsub:isotropic-forms}

Set
\begin{equation}
R=\nu^2\Id,
\qquad
\nu>0.
\label{eq:appB-isotropic-R}
\end{equation}
Then
\begin{equation}
M_y=
[1+(\nu^2-1)y]\Id,
\end{equation}
and the two integrals can be evaluated elementarily:
\begin{align}
f_0(\nu^2\Id)
&=
\int_0^1
\frac{y\,\dd y}
{[1+(\nu^2-1)y]^{3/2}}
=
\frac{2}{\nu(\nu+1)^2},
\label{eq:appB-f0-isotropic}\\
F(\nu^2\Id)
&=
\int_0^1
\frac{y\nu^2\,\dd y}
{[1+(\nu^2-1)y]^{5/2}}\Id
\nonumber\\
&=
\frac{2(2\nu+1)}
{3\nu(\nu+1)^2}\Id.
\label{eq:appB-F-isotropic}
\end{align}
Therefore
\begin{equation}
\mathcal H(\nu^2\Id)
=
h(\nu^2)\Id,
\qquad
h(\nu^2)
=
\frac{4(\nu-1)}
{3\nu(\nu+1)^2}.
\label{eq:appB-h-isotropic}
\end{equation}
In particular,
\begin{equation}
\operatorname{sgn}h(\nu^2)=\operatorname{sgn}(\nu-1),
\end{equation}
so the fermion cone moves toward the scalar cone, while
eq.~\eqref{eq:appB-Bflow} moves the scalar cone toward the fermion cone.
The resulting attraction for an arbitrary positive velocity ratio is
consistent with the known one-loop Yukawa velocity flow
\cite{AnberDonoghue2011,RoyJuricicHerbut2016}.

At the common cone \(\nu=1\),
\begin{equation}
f_0(\Id)=\frac12,
\qquad
F(\Id)=\frac12\Id,
\qquad
\mathcal H(\Id)=0,
\label{eq:appB-commoncone-values}
\end{equation}
as required by common-cone stationarity.

\subsection{Fr\'echet derivative and common-cone linearization}
\label{appsub:Frechet-H}

The matrix derivative needed in sections~\ref{subsec:endpoint-propagation}
and~\ref{subsec:near-sync} follows directly from
eq.~\eqref{eq:appB-H}.  Since
\begin{equation}
D M_y(R)[K]=yK,
\end{equation}
we have
\begin{equation}
D(M_y^{-1})[K]
=
-yM_y^{-1}KM_y^{-1},
\label{eq:appB-DMinv}
\end{equation}
and
\begin{equation}
D(\det M_y)^{-1/2}[K]
=
-\frac{y}{2}
(\det M_y)^{-1/2}
\Tr(M_y^{-1}K).
\label{eq:appB-Ddet}
\end{equation}
Consequently
\begin{align}
D\mathcal H_R[K]
={}&
\int_0^1\dd y\,
\frac{y(1-y)}{\sqrt{\det M_y}}
\Bigg[
M_y^{-1}KM_y^{-1}
\nonumber\\
&\hspace{18mm}-
\frac12
\left(\Id-M_y^{-1}\right)
\Tr(M_y^{-1}K)
\Bigg].
\label{eq:appB-DH}
\end{align}

At \(R=\Id\), the second term vanishes and
\begin{equation}
D\mathcal H_{\Id}[K]
=
\left(
\int_0^1 y(1-y)\,\dd y
\right)K
=
\frac16K.
\label{eq:appB-DH-I}
\end{equation}
Thus
\begin{equation}
\mathcal H(\Id+K)
=
\frac16K+O(\|K\|^2),
\label{eq:appB-H-linearization}
\end{equation}
which is the coefficient used in the near-synchronization consensus system.

There is also a useful projected derivative at an equality direction.
If
\begin{equation}
Rz=z,
\qquad
\|z\|=1,
\label{eq:appB-unit-eigenvector}
\end{equation}
then
\begin{equation}
M_yz=z,
\qquad
z^T(\Id-M_y^{-1})z=0.
\end{equation}
Taking the \(z\)-expectation of
eq.~\eqref{eq:appB-DH} therefore removes the determinant-derivative term and
gives
\begin{equation}
z^T(D\mathcal H_R[K])z
=
c(R)\,z^TKz,
\label{eq:appB-projected-DH}
\end{equation}
where
\begin{equation}
c(R)
=
\int_0^1
\frac{y(1-y)\,\dd y}
{\sqrt{\det[(1-y)\Id+yR]}}
>0.
\label{eq:appB-cR}
\end{equation}
This is the positive transfer coefficient used in the endpoint-propagation
argument.  Appendix~\ref{app:global-proofs} applies
eq.~\eqref{eq:appB-projected-DH} to the full Taylor-order induction.

\section{Proofs of global nonexpansion and finite-window contraction}
\label{app:global-proofs}

This appendix gives the technical proofs underlying
Theorems~\ref{thm:global-nonexpansion} and~\ref{thm:finite-window}.  We first
treat degenerate active generalized eigenvalues without differentiating a
matrix square root, then prove the extremal sign and equality statements in
their order-theoretic form.  The second half establishes finite-order propagation of strictness, handles
branch switching of the diameter envelope, and proves uniform contraction
factors away from and near the synchronized manifold.

\subsection{Generalized-eigenvalue Dini derivatives}
\label{appsub:pencil-dini}

Let \(A(s),B(s)\in\SPD(3)\) be \(C^1\) matrix curves and define
\begin{equation}
\lambda_{\max}(A,B)
=
\max_{x\neq0}
\frac{x^TBx}{x^TAx}.
\label{eq:appC-Rayleigh}
\end{equation}
At a fixed time \(s_0\), write
\begin{equation}
\lambda=\lambda_{\max}(A,B),
\qquad
\mathcal E
=
\{x:\;Bx=\lambda Ax\}.
\label{eq:appC-active-space}
\end{equation}
Normalize active vectors by \(x^TAx=1\).  The upper Dini derivative is
\begin{equation}
D^+\lambda_{\max}(A,B)
=
\max_{\substack{x\in\mathcal E\\x^TAx=1}}
x^T(\dot B-\lambda\dot A)x.
\label{eq:appC-pencil-Dini}
\end{equation}

A short proof follows directly from the generalized Rayleigh quotient.
For a fixed nonzero \(x\),
\begin{equation}
r_x(s)=\frac{x^TB(s)x}{x^TA(s)x}
\end{equation}
has derivative, at an active normalized vector,
\begin{equation}
\dot r_x
=
x^T(\dot B-\lambda\dot A)x.
\label{eq:appC-Rayleigh-derivative}
\end{equation}
The maximum in eq.~\eqref{eq:appC-Rayleigh} is taken over the compact
\(A\)-unit sphere.  The standard directional-derivative formula for the
maximum of a smooth family of functions therefore gives
eq.~\eqref{eq:appC-pencil-Dini}; in finite dimensions it can also be obtained
by taking maximizing vectors along a sequence \(s_n\downarrow s_0\) and
extracting a convergent subsequence.  This argument treats a degenerate
largest generalized eigenvalue without choosing an eigenbasis inside
\(\mathcal E\).

For the network define
\begin{equation}
\Lambda_{ij}
=
\lambda_{\max}(G_i,G_j),
\qquad
\Lambda=\max_{i,j}\Lambda_{ij}.
\label{eq:appC-network-lambda}
\end{equation}
Since the maximum is over finitely many ordered pairs,
\begin{equation}
D^+\Lambda
\le
\max_{(i,j)\in\mathcal A}
D^+\Lambda_{ij},
\qquad
\mathcal A=\{(i,j):\Lambda_{ij}=\Lambda\}.
\label{eq:appC-finite-max-Dini}
\end{equation}
Equations~\eqref{eq:appC-pencil-Dini} and
\eqref{eq:appC-finite-max-Dini} justify the degenerate-active-pair step in
section~\ref{subsec:global-nonexpansion}.

\subsection{Rank-one sign, equality rigidity, and endpoint signs}
\label{appsub:rankone-proof}

We give a self-contained proof of the kernel inequality used by
Theorem~\ref{thm:global-nonexpansion}.  Let \(R\in\SPD(3)\) and let \(u\) be
a Euclidean unit vector.  We write \(X\succeq Y\) for Loewner order,
meaning that \(X-Y\) is positive semidefinite.  The rank-one order relation
\begin{equation}
R\succeq uu^T
\label{eq:appC-rankone-order}
\end{equation}
is equivalent to
\begin{equation}
u^TR^{-1}u\le1.
\label{eq:appC-Schur}
\end{equation}
Indeed, after congruence by \(R^{-1/2}\),
eq.~\eqref{eq:appC-rankone-order} becomes
\begin{equation}
\Id\succeq vv^T,
\qquad
v=R^{-1/2}u,
\end{equation}
which is equivalent to \(\|v\|^2\le1\), with \(\|\cdot\|\) the Euclidean
vector norm in this paragraph.

For
\begin{equation}
M_y=(1-y)\Id+yR,
\end{equation}
operator convexity of inversion gives
\begin{equation}
M_y^{-1}
\preceq
(1-y)\Id+yR^{-1}.
\label{eq:appC-operator-convex}
\end{equation}
Therefore
\begin{equation}
u^T(\Id-M_y^{-1})u
\ge
y(1-u^TR^{-1}u)
\ge0.
\label{eq:appC-integrand-sign}
\end{equation}
Multiplying by the positive weight
\((1-y)/\sqrt{\det M_y}\) and integrating yields
\begin{equation}
R\succeq uu^T
\quad\Longrightarrow\quad
u^T\mathcal H(R)u\ge0.
\label{eq:appC-rankone-sign}
\end{equation}

The equality case is rigid.  If
\(u^T\mathcal H(R)u=0\), the continuous nonnegative integrand in
eq.~\eqref{eq:appC-integrand-sign} vanishes for every \(y\in[0,1]\).
Expanding \(M_y^{-1}\) at \(y=0\),
\begin{equation}
M_y^{-1}
=
\Id-y(R-\Id)+O(y^2),
\end{equation}
gives
\begin{equation}
u^TRu=1.
\label{eq:appC-uRu}
\end{equation}
Since \(R-uu^T\succeq0\),
\begin{equation}
u^T(R-uu^T)u=0
\end{equation}
implies
\begin{equation}
(R-uu^T)u=0,
\end{equation}
and hence
\begin{equation}
Ru=u.
\label{eq:appC-Ru-u}
\end{equation}
The converse is immediate because \(M_yu=u\) for every \(y\).  Thus
\begin{equation}
u^T\mathcal H(R)u=0
\quad\Longleftrightarrow\quad
Ru=u
\qquad
(R\succeq uu^T).
\label{eq:appC-rigidity}
\end{equation}

The inversion identity
\begin{equation}
\mathcal H(R^{-1})
=
-\sqrt{\det R}\,\mathcal H(R)
\label{eq:appC-inversion}
\end{equation}
converts eq.~\eqref{eq:appC-rankone-sign} into the dual statement
\begin{equation}
R^{-1}\succeq uu^T
\quad\Longrightarrow\quad
u^T\mathcal H(R)u\le0,
\label{eq:appC-dual-sign}
\end{equation}
with equality again equivalent to \(Ru=u\).

We now spell out the endpoint equality used later.  Let
\((i,j)\) be active,
\begin{equation}
G_jx=\Lambda G_ix,
\qquad
x^TG_ix=1.
\label{eq:appC-active-pair}
\end{equation}
At a lower bosonic endpoint \(G_i=B_r\), every neighboring fermion satisfies
\begin{equation}
x^TA_ax\ge x^TB_rx=1.
\end{equation}
If the lower endpoint derivative vanishes, positivity of all edge
contributions forces
\begin{equation}
x^TA_ax=1
\qquad
\text{for every }a\sim r.
\label{eq:appC-boson-equality-scalar}
\end{equation}
Since
\begin{equation}
G_j\preceq\Lambda A_a
\end{equation}
and
\(x^T(\Lambda A_a-G_j)x=0\), positivity of the matrix
\(\Lambda A_a-G_j\) implies
\begin{equation}
A_ax=B_rx.
\label{eq:appC-boson-vector-equality}
\end{equation}
At a lower fermionic endpoint \(G_i=A_a\), equality in
eq.~\eqref{eq:appC-rankone-sign} gives exactly the same relation
\(B_rx=A_ax\) for every incident edge.  The upper endpoint follows by the
dual argument.  Thus:

\begin{lemma}[Endpoint equality propagation]
\label{lem:appC-endpoint-equality}
If an active lower or upper endpoint has zero directional derivative along
an active generalized eigenvector, then the vector-valued cone relation
\(G_vx\) agrees across every incident edge.
\end{lemma}

For a complete bipartite graph this already implies pointwise strictness.
If both active endpoint derivatives vanished in a nonsynchronized state, the
relation in Lemma~\ref{lem:appC-endpoint-equality} would propagate through a
direct edge when the endpoints lie in opposite partitions, or through any
shared neighbor when they lie in the same partition.  It would follow that
\(G_jx=G_ix\), contradicting
\(G_jx=\Lambda G_ix\) with \(\Lambda>1\).

\subsection{Finite-order propagation of strictness}
\label{appsub:jet-propagation}

For a sparse connected graph, endpoint equality can persist through several
edges.  Finite graph distance nevertheless forces strictness to appear at a
finite Taylor order.

Let \(s=0\) be a reference time and let
\(\lambda(s)\) be a real-analytic active generalized-eigenvalue branch with
\begin{equation}
\lambda(0)=\Lambda>1.
\end{equation}
Choose a corresponding analytic generalized-eigenvector branch \(x(s)\)
inside a local Rellich block and normalize it by
\begin{equation}
x(s)^TG_i(s)x(s)=1
\label{eq:appC-x-normalization}
\end{equation}
at the lower endpoint.  The active-branch logarithm is
\begin{equation}
D(s)=\log\lambda(s).
\end{equation}

Assume the first derivative vanishes.  By
Lemma~\ref{lem:appC-endpoint-equality}, the directional cone vector agrees
across every incident edge at the endpoint.  Choose a shortest graph path
\begin{equation}
v_0=i,\ v_1,\ldots,v_m=j
\label{eq:appC-transfer-path}
\end{equation}
from this lower endpoint to the active upper endpoint, and set
\begin{equation}
r_k=x(0)^TG_{v_k}(0)x(0).
\label{eq:appC-path-rk}
\end{equation}
The global extremality of the active pair gives, for every network node
\(v\),
\begin{equation}
G_j(0)\preceq\Lambda G_v(0),
\qquad
G_v(0)\preceq\Lambda G_i(0).
\label{eq:appC-global-path-orders}
\end{equation}
Using eqs.~\eqref{eq:appC-active-pair} and
\eqref{eq:appC-x-normalization}, these inequalities imply
\begin{equation}
1\le x(0)^TG_v(0)x(0)\le\Lambda.
\label{eq:appC-path-scalar-bounds}
\end{equation}
Moreover, saturation of the lower bound is rigid.  If
\(x^TG_vx=1\) at \(s=0\), then the positive-semidefinite matrix
\(\Lambda G_v-G_j\) has zero quadratic form on \(x\), and therefore
\begin{equation}
(\Lambda G_v-G_j)x=0.
\end{equation}
Together with \(G_jx=\Lambda G_ix\), this gives
\begin{equation}
G_vx=G_ix.
\label{eq:appC-lower-saturation-vector}
\end{equation}
Thus, along the chosen path, lower scalar saturation implies the full
propagated vector-valued equality.  Since \(r_0=1\) and \(r_m=\Lambda>1\),
there is a first index \(p\) for which \(r_p>1\).  For every \(k<p\),
eq.~\eqref{eq:appC-lower-saturation-vector} gives
\(G_{v_k}x=G_ix\), while the first departing edge
\(e_*=(v_{p-1},v_p)\) satisfies
\begin{equation}
x^T\bigl(G_{v_p}-G_{v_{p-1}}\bigr)x>0.
\label{eq:appC-first-scalar-departure}
\end{equation}
If \(v_{p-1}\) is bosonic, eq.~\eqref{eq:appC-boson-flux} therefore gives a
strictly positive lower-endpoint quadratic flux.  If \(v_{p-1}\) is
fermionic, use its equality frame and the same global-order argument as in
Lemma~\ref{lem:appC-endpoint-equality}: one has
\(R\succeq uu^T\), while eq.~\eqref{eq:appC-first-scalar-departure} excludes
\(Ru=u\).  The rigidity equivalence
\eqref{eq:appC-rigidity} then yields
\(u^T\mathcal H(R)u>0\).  Hence in either endpoint type the first edge
leaving the lower saturated plateau supplies a strictly inward order-one
quadratic flux.  We transfer this coefficient back toward the active
endpoint through the preceding equality edges.

For a bosonic update the edge vector field is linear,
\begin{equation}
\mathcal F^{(B)}_e
=
\widehat\alpha_e(A-B).
\label{eq:appC-boson-flux}
\end{equation}
Suppose the edge mismatch and all of its directional derivatives through
order \(m-1\) vanish at the preceding equality edge, while the first
nonzero inward coefficient appears at order \(m\) on the adjacent node.
The product rule gives, at order \(m+1\),
\begin{align}
\partial_s^{m+1}\mathcal F^{(B)}_e(0)
={}&
\widehat\alpha_e(0)
\partial_s^{m+1}(A-B)(0)
\nonumber\\
&+
\sum_{r=1}^{m+1}
\binom{m+1}{r}
\partial_s^r\widehat\alpha_e(0)
\nonumber\\
&\hspace{18mm}\times
\partial_s^{m+1-r}(A-B)(0).
\label{eq:appC-boson-product}
\end{align}
Every term in the sum contains a lower-order mismatch derivative and
therefore vanishes by the induction hypothesis.  The leading coefficient is
multiplied only by
\(\widehat\alpha_e(0)>0\).

For a fermionic update, return to the equality frame of the edge.  The
zeroth-order relative matrix satisfies
\begin{equation}
Rz=z.
\end{equation}
The first variation of the Yukawa kernel projected onto \(z\) is, by
eqs.~\eqref{eq:appB-projected-DH}--\eqref{eq:appB-cR},
\begin{equation}
z^T(D\mathcal H_R[K])z
=
c(R)z^TKz,
\qquad
c(R)>0.
\label{eq:appC-Frechet-transfer}
\end{equation}
Thus the first nonzero coefficient transmitted through a fermionic equality
edge is multiplied by
\begin{equation}
2\widehat\beta_e(0)c(R)>0.
\label{eq:appC-fermion-multiplier}
\end{equation}
Higher derivatives of \(\widehat\beta_e\), of the congruence frame, and of
the active eigenvector multiply lower-order edge mismatches.  Those
mismatches vanish at the induction step, so such terms do not alter the
first nonzero order or its sign.

Induction along the finite path proves:

\begin{lemma}[Finite-order endpoint propagation]
\label{lem:appC-finite-order}
For a nonsynchronized analytic connected trajectory, every analytic active
branch with value \(\lambda>1\) has a finite Taylor order \(q\ge1\) at which
it moves strictly inward:
\begin{equation}
D(s)=D(0)-Cs^q+O(s^{q+1}),
\qquad
C>0.
\label{eq:appC-finite-order-expansion}
\end{equation}
When the first derivative vanishes, the order is at most the number of edges
from the active endpoint to the first edge that leaves the lower saturated
plateau \(x^TG_vx=1\) along the chosen path.  By
eq.~\eqref{eq:appC-lower-saturation-vector}, all preceding edges satisfy the
full propagated vector equality.  Rate and frame derivatives do not change
the leading order.
\end{lemma}

For general data the lemma establishes finite-order strictness.  If several
shortest transfer routes reach the same active endpoint, every leading edge
multiplier in eqs.~\eqref{eq:appC-boson-product} and
\eqref{eq:appC-fermion-multiplier} is positive.  Their first surviving
coefficients therefore add with the same inward sign; the explicit path-sum
form is given in appendix~\ref{appsub:topology-induction}.  The two-plateau
structure of section~\ref{sec:topology-onset} identifies the resulting order
with the graph-distance quantity \(q_{ij}\).

\subsection{Analytic branches and exclusion of constant-diameter intervals}
\label{appsub:branch-switching}

The extended normalized RG vector field is real analytic on the compact
controlled domain.  Hence every matrix entry \(G_i(s)\) is real analytic in
\(s\).  For an ordered pair \((i,j)\), first pass from the generalized
symmetric pencil \((G_j,G_i)\) to the ordinary symmetric family
\begin{equation}
C_{ij}(s)
=
G_i(s)^{-1/2}G_j(s)G_i(s)^{-1/2}
\label{eq:appC-Cij}
\end{equation}
is analytic because the principal inverse square root is analytic on
\(\SPD(3)\).  This reduction includes degenerate generalized eigenvalues:
Rellich's theorem is applied to the full analytic symmetric block, not to a
preselected simple eigenvector.  Rellich's theorem for analytic
real-symmetric matrix families \cite{Kato1995} therefore supplies finitely
many real-analytic local eigenvalue branches
\begin{equation}
\lambda_{ij,1}(s),\ldots,\lambda_{ij,3}(s).
\label{eq:appC-Rellich-branches}
\end{equation}
The network extremum is their finite upper envelope,
\begin{equation}
\Lambda(s)
=
\max_{i,j,n}\lambda_{ij,n}(s).
\label{eq:appC-envelope}
\end{equation}

Suppose a connected nonsynchronized trajectory had
\begin{equation}
\Lambda(s)=\Lambda_*>1
\end{equation}
through a nonempty open interval \(I\).  Restrict to a smaller interval on
which all branches in eq.~\eqref{eq:appC-Rellich-branches} are defined.
For each branch set
\begin{equation}
Z_{ij,n}
=
\{s\in I:\lambda_{ij,n}(s)=\Lambda_*\}.
\end{equation}
The finitely many closed sets \(Z_{ij,n}\) cover \(I\).  By the Baire
category theorem, one of them contains a nonempty open subinterval.
Equivalently, a real-analytic function that is not identically zero has only
isolated zeros.  Hence at least one analytic branch equals \(\Lambda_*\) on
a subinterval.  This contradicts
Lemma~\ref{lem:appC-finite-order}.  Therefore a nonsynchronized connected
trajectory cannot have constant Thompson diameter on any nonzero
accumulated-time interval.

Since Theorem~\ref{thm:global-nonexpansion} already gives monotonicity, we
have for every nonsynchronized state and every \(T>0\),
\begin{equation}
\DT(s+T)<\DT(s)
\label{eq:appC-strict-window}
\end{equation}
whenever the flow exists on the entire window.

\subsection{Uniform contraction away from synchronization}
\label{appsub:away-uniform}

Fix a reference node \(v_0\).  A smooth representative of the relative
congruence class is obtained by
\begin{equation}
\widetilde G_i
=
G_{v_0}^{-1/2}G_iG_{v_0}^{-1/2},
\qquad
\widetilde G_{v_0}=\Id.
\label{eq:appC-relative-section}
\end{equation}
If
\(\DT\le D_0\), then
\begin{equation}
e^{-D_0}\Id
\preceq
\widetilde G_i
\preceq
e^{D_0}\Id.
\label{eq:appC-compact-section}
\end{equation}
Compactness of the remaining variables \(Z\), together with the normalized
rate bounds, defines a compact normalized state space \(\mathcal K\).

For \(\varepsilon>0\), let
\begin{equation}
\mathcal K_\varepsilon
=
\{\mathsf X\in\mathcal K:\DT(\mathsf X)\ge\varepsilon\}.
\label{eq:appC-Keps}
\end{equation}
Let \(\Phi_T\) denote the time-\(T\) map of the normalized autonomous flow.
ODE continuous dependence makes \(\Phi_T\) continuous on
\(\mathcal K_\varepsilon\), and Thompson distance is continuous on
\(\SPD(3)\times\SPD(3)\).  Therefore
\begin{equation}
R_T(\mathsf X)
=
\frac{\DT(\Phi_T\mathsf X)}{\DT(\mathsf X)}
\label{eq:appC-ratio}
\end{equation}
is continuous on \(\mathcal K_\varepsilon\).  By
eq.~\eqref{eq:appC-strict-window},
\(R_T(\mathsf X)<1\) pointwise.  Compactness then gives
\begin{equation}
\rho^{\rm out}_{T,\varepsilon}
=
\max_{\mathsf X\in\mathcal K_\varepsilon}R_T(\mathsf X)
<1.
\label{eq:appC-rho-out}
\end{equation}

\subsection{Uniform contraction near the synchronized manifold}
\label{appsub:near-uniform}

We now control the region in which the relative diameter is small.  By one
common congruence, choose local coordinates around a synchronized state in
which the common cone is the identity:
\begin{equation}
A_a=\Id+X_a,
\qquad
B_r=\Id+Y_r.
\label{eq:appC-near-coordinates}
\end{equation}
Only differences between the matrices are transverse to the synchronized
manifold.  Define the operator-norm oscillation
\begin{equation}
\operatorname{osc}_{\rm op}(X,Y)
=
\max_{u,v}
\|Z_u-Z_v\|_{\rm op},
\label{eq:appC-op-osc}
\end{equation}
where \(Z_u\) denotes any \(X_a\) or \(Y_r\).  This is a norm on the quotient
of the finite family of symmetric matrices by the common diagonal subspace.

Appendix~\ref{app:selfenergies} gives
\begin{equation}
D\mathcal H_{\Id}[K]=\frac16K,
\end{equation}
and hence the transverse linearization
\begin{subequations}
\label{eq:appC-linear-consensus}
\begin{align}
\dot X_a
&=
\sum_{r\sim a}
\frac{\widehat\beta_{ar}}{3}(Y_r-X_a),
\\
\dot Y_r
&=
\sum_{a\sim r}
\widehat\alpha_{ar}(X_a-Y_r).
\end{align}
\end{subequations}
For each matrix entry this is the same time-dependent scalar consensus
system
\begin{equation}
\dot z=L(s)z,
\label{eq:appC-scalar-consensus}
\end{equation}
where \(L(s)\) is Metzler, has zero row sums, and is irreducible on the fixed
connected bidirected graph.  If
\begin{equation}
m_0=\frac{m}{3},
\qquad
M_0=M,
\label{eq:appC-m0M0}
\end{equation}
then every directed edge rate lies in
\([m_0,M_0]\) after using the smaller lower bound for the two orientations.

Let \(N\) be the total number of nodes and
\(\Delta\) the maximal graph degree.  Denote the transition matrix over a
window \(T\) by \(P(s+T,s)\).  It is row stochastic.  Positivity and
variation of constants along a directed path of length \(d\le N-1\) give
the uniform entry bound
\begin{equation}
P_{uv}(s+T,s)
\ge
e^{-\Delta M_0T}
\frac{(m_0T)^d}{d!}.
\label{eq:appC-path-lower-bound}
\end{equation}
For \(u=v\), the same estimate holds with \(d=0\).  Consequently every entry
is bounded below by
\begin{equation}
\eta_T
=
e^{-\Delta M_0T}
\min_{0\le d\le N-1}
\frac{(m_0T)^d}{d!}
>0.
\label{eq:appC-eta}
\end{equation}

For a row-stochastic matrix \(P\), define the Dobrushin coefficient
\begin{equation}
\delta(P)
=
1-
\min_{u,v}
\sum_w\min(P_{uw},P_{vw}).
\label{eq:appC-Dobrushin}
\end{equation}
Equation~\eqref{eq:appC-eta} gives
\begin{equation}
\delta(P(s+T,s))
\le
1-N\eta_T
\equiv
\delta_T
<1.
\label{eq:appC-Dobrushin-bound}
\end{equation}
For any norm on the value space, convexity yields
\begin{equation}
\max_{u,v}
\left\|
\sum_w(P_{uw}-P_{vw})Z_w
\right\|
\le
\delta(P)
\max_{u,v}\|Z_u-Z_v\|.
\label{eq:appC-Dobrushin-anynorm}
\end{equation}
Taking the operator norm shows that the linearized evolution contracts
\(\operatorname{osc}_{\rm op}\) by at most \(\delta_T\).

It remains to compare this quotient norm with Thompson diameter.  For
\(\|Z_u\|_{\rm op}\le\epsilon\) with \(\epsilon<1/2\), the logarithmic
generalized-eigenvalue expansion gives uniformly
\begin{equation}
d_T(\Id+Z_u,\Id+Z_v)
=
\|Z_v-Z_u\|_{\rm op}
+
O\!\left(
\epsilon\|Z_v-Z_u\|_{\rm op}
\right).
\label{eq:appC-local-Thompson}
\end{equation}
Hence there are constants \(c_-(\epsilon),c_+(\epsilon)\to1\) as
\(\epsilon\to0\) such that
\begin{equation}
c_-\operatorname{osc}_{\rm op}
\le
\DT
\le
c_+\operatorname{osc}_{\rm op}.
\label{eq:appC-local-equivalence}
\end{equation}

Because the full normalized vector field is analytic and vanishes on the
synchronized manifold in the relative directions, its transverse nonlinear
remainder obeys
\begin{equation}
\|\mathcal N(Z,s)\|_{\rm osc}
\le
C\epsilon\,\|Z\|_{\rm osc}
\label{eq:appC-nonlinear-remainder}
\end{equation}
throughout a sufficiently small relative neighborhood, uniformly over the
compact common-mode and \(Z\)-variables.  Theorem~\ref{thm:global-nonexpansion}
keeps a trajectory that starts in this neighborhood from leaving it.
Variation of constants over the fixed window \(T\), together with
eq.~\eqref{eq:appC-Dobrushin-bound}, gives
\begin{equation}
\operatorname{osc}_{\rm op}(s+T)
\le
\left(
\delta_T+C_T\epsilon
\right)
\operatorname{osc}_{\rm op}(s),
\label{eq:appC-nonlinear-osc-contraction}
\end{equation}
where \(C_T\) is uniform on the compact domain.  Choose \(\epsilon\) so small
that
\begin{equation}
\delta_T+C_T\epsilon<1.
\label{eq:appC-epsilon-choice}
\end{equation}
Combining with eq.~\eqref{eq:appC-local-equivalence} and, if necessary,
shrinking \(\epsilon\) once more, gives
\begin{equation}
\DT(s+T)
\le
\rho_T^{\rm in}\DT(s),
\qquad
0<\rho_T^{\rm in}<1.
\label{eq:appC-rho-in}
\end{equation}

\subsection{Patching and connected components}
\label{appsub:patching}

Choose the same small relative threshold \(\varepsilon\) used in the near
region.  Equations~\eqref{eq:appC-rho-out} and
\eqref{eq:appC-rho-in} give
\begin{equation}
\rho_T
=
\max
\left\{
\rho^{\rm out}_{T,\varepsilon},
\rho_T^{\rm in}
\right\}
<1.
\label{eq:appC-rho-patch}
\end{equation}
Therefore every controlled connected trajectory satisfies
\begin{equation}
\DT(s+T)\le\rho_T\DT(s).
\label{eq:appC-window-final}
\end{equation}
Iteration yields
\begin{equation}
\DT(s+nT)
\le
\rho_T^n\DT(s),
\qquad
n\in\mathbb N.
\label{eq:appC-iteration}
\end{equation}
If \(s(\ell)\to\infty\), this proves
\begin{equation}
\DT(\ell)\to0.
\end{equation}

If the interaction graph has several connected components, the proof applies
to each component separately.  Thus every component synchronizes internally
under its own persistence conditions.  For a graph with \(k\) connected
components, the infrared theory has at most \(k\) mutually unlocked relative
cone sectors.

A useful exact special case occurs when all edge rates share a common
positive factor,
\begin{equation}
\alpha_e(\ell)=\chi(\ell)\bar\alpha_e,
\qquad
\beta_e(\ell)=\chi(\ell)\bar\beta_e,
\end{equation}
with fixed positive barred rates.  The time change
\begin{equation}
s(\ell)=\int_0^\ell\chi(u)\,\dd u
\end{equation}
reduces the flow exactly to the fixed-rate normalized system.

\section{Persistence, topology, and numerical details}
\label{app:persistence-topology}

This appendix completes the field-theoretic persistence argument of
section~\ref{subsec:qft-persistence}, gives the detailed Taylor-order proof
behind Theorem~\ref{thm:topology-onset}, and records the numerical
implementation used in Figs.~\ref{fig:global-locking} and
\ref{fig:topology-onset}.

\subsection{Analyticity of anisotropic one-loop coupling coefficients}
\label{appsub:anisotropic-vertex}

We first establish the analytic dependence used in the finite-mismatch
persistence bootstrap.  Consider the Pauli-string Yukawa subclass of
appendix~\ref{app:internal-algebra}.  At one loop, the logarithmically
divergent local part of an external Yukawa channel \(r\) is a sum of
wave-function terms and triangle terms with internal channel \(s\).  The
internal flavor structures are
\begin{equation}
T_s^2T_r,\qquad
T_rT_s^2,\qquad
T_sT_rT_s.
\label{eq:appD-flavor-structures}
\end{equation}
For Pauli strings,
\begin{equation}
T_sT_rT_s=\eta_{sr}T_r,
\qquad
\eta_{sr}=\pm1,
\label{eq:appD-Pauli-conjugation}
\end{equation}
so the flavor direction is preserved.  The anisotropic kinetic matrices
enter through the scalar coefficient multiplying that direction.

To see the regularity of this coefficient, whiten the fermion spatial
quadratic form and denote the relevant relative scalar matrices by
\(R_s\in\SPD(3)\).  After Feynman parametrization, every logarithmically
divergent triangle coefficient can be reduced to a finite sum of integrals
of the form
\begin{equation}
C_{rs}(G)
=
\int_{\Delta}
\dd\boldsymbol x\,
\frac{
P_{rs}\!\left(
\boldsymbol x,
M(\boldsymbol x;G)^{-1},
G
\right)
}{
\sqrt{\det M(\boldsymbol x;G)}
},
\label{eq:appD-Crs-generic}
\end{equation}
where \(G\) collectively denotes the propagation matrices entering the
loop, \(\Delta\) is a compact Feynman-parameter simplex,
\(\boldsymbol x\) denotes its parameter vector, \(P_{rs}\) is polynomial
in its matrix arguments, and \(M(\boldsymbol x;G)\) is a positive affine
combination of the spatial
quadratic forms entering the propagators.  For example, the triangle with
two identical fermion denominators and one scalar denominator has, in the
whitened frame,
\begin{equation}
M_x=(1-x)\Id+xR_s
\label{eq:appD-triangle-M}
\end{equation}
after combining the two fermion parameters.

Let \(\lambda_{\min}(X)\) and \(\lambda_{\max}(X)\) denote the smallest and
largest eigenvalues of an SPD matrix \(X\).  Let \(\mathcal K_G\) be a
compact SPD domain satisfying
\begin{equation}
0<\lambda_-
\le
\lambda_{\min}(G_v)
\le
\lambda_{\max}(G_v)
\le
\lambda_+<\infty
\label{eq:appD-SPD-compact}
\end{equation}
for every propagation matrix.  Every affine matrix
\(M(\boldsymbol x;G)\) then obeys a uniform positive lower bound on
\(\Delta\times\mathcal K_G\).  Hence
\begin{equation}
M^{-1},
\qquad
(\det M)^{-1/2}
\label{eq:appD-analytic-building-blocks}
\end{equation}
are real analytic functions of the kinetic-matrix entries and are uniformly
bounded together with all derivatives on compact subsets.  Differentiation
under the finite parameter integral in
eq.~\eqref{eq:appD-Crs-generic} is therefore justified to every order.

\begin{lemma}[Analytic anisotropic coupling coefficients]
\label{lem:appD-analytic-vertex}
On every compact positive-definite no-tilt kinetic domain, the logarithmic
one-loop Yukawa-vertex and scalar-quartic coefficients are real analytic
functions of the propagation matrices.  The Yukawa counterterms remain in
the original Pauli-string channel subspace, while the
\(\mathbb Z_2\times\mathbb Z_2\)-even quartic sector
\((\lambda_1,\lambda_2,\lambda_{12})\) remains closed.
\end{lemma}

The Dirac part of the same statement follows from symmetry of the loop
momentum tensors.  After the Gaussian or angular integrations, every
quadratic loop-momentum tensor is symmetric.  Its contraction with products of spatial gamma matrices therefore projects
onto the anticommutator part, preserving the original Dirac Yukawa vertex
structure.  Together
with the two-point closure proved in appendices~\ref{app:internal-algebra}
and~\ref{app:selfenergies}, this gives a closed analytic one-loop vector
field for the couplings and propagation matrices on the compact domain.

The analyticity statement also supplies the uniform finite-mismatch estimate
used in section~\ref{subsec:qft-persistence}.  Let \(G_{\rm c}\) denote the
common-cone point obtained by setting all relative propagation matrices equal
inside a fixed common-congruence section, and let
\(\boldsymbol p=(q,u,v,w)\) collect the normalized coupling variables of
eq.~\eqref{eq:appA-normalized-quartic-defs}.  Denote the corresponding
normalized coupling vector field by
\(\boldsymbol F(G,\boldsymbol p)\).  In this subsection,
\(\|\cdot\|\) denotes the matrix operator norm for propagation matrices,
the Euclidean norm for the finite-dimensional coupling vector, and the
corresponding induced norm for linear maps.  On a compact SPD domain,
Thompson distance and the matrix operator norm are locally equivalent, so
there is a finite
constant \(C_G\) such that
\begin{equation}
\|G-G_{\rm c}\|
\le
C_G\,\DT(G)
\label{eq:appD-local-metric-equivalence}
\end{equation}
throughout a sufficiently small common-cone neighborhood.  Let \(D_G\boldsymbol F\) denote the Fr\'echet derivative of
\(\boldsymbol F\) with respect to its propagation-matrix arguments.  The
mean-value identity
\begin{align}
\boldsymbol F(G,\boldsymbol p)
-
\boldsymbol F(G_{\rm c},\boldsymbol p)
={}&
\int_0^1
D_G\boldsymbol F
\!\left(
G_{\rm c}+t(G-G_{\rm c}),
\boldsymbol p
\right)
\nonumber\\
&\times [G-G_{\rm c}]\,\dd t.
\label{eq:appD-mean-value}
\end{align}
then gives
\begin{equation}
\left\|
\boldsymbol F(G,\boldsymbol p)
-
\boldsymbol F(G_{\rm c},\boldsymbol p)
\right\|
\le
C_{\rm an}\,\DT(G),
\label{eq:appD-Lipschitz}
\end{equation}
where \(C_{\rm an}<\infty\) is the supremum of this derivative norm on the
compact interpolation set.

\subsection{Finite-mismatch trapping of the quartic sector}
\label{appsub:quartic-trapping}

We now use eq.~\eqref{eq:appD-Lipschitz} to extend the common-cone sink
\(p_*\) of appendix~\ref{appsub:quartic-completion} to finite cone mismatch.
Let \(\xi\) be the deviation vector in eq.~\eqref{eq:appA-xi}; a prime
continues to denote \(\dd/\dd\vartheta\), with \(\vartheta\) defined in
eq.~\eqref{eq:appA-vartheta}.  Near
\(p_*\), the common-cone normalized vector field has the form
\begin{equation}
\xi'
=
J_*\xi
+
\boldsymbol N(\xi),
\qquad
\|\boldsymbol N(\xi)\|
\le
C_N\|\xi\|^2,
\label{eq:appD-common-quartic-normal-form}
\end{equation}
for some finite \(C_N\).  At finite mismatch,
eq.~\eqref{eq:appD-Lipschitz} adds a remainder
\(\boldsymbol R_G\) satisfying
\begin{equation}
\|\boldsymbol R_G\|
\le
C_R\,\DT(G)
\label{eq:appD-quartic-remainder}
\end{equation}
on a sufficiently small compact neighborhood.

All eigenvalues of \(J_*\) have negative real part by
eq.~\eqref{eq:appA-quartic-eigenvalues}.  Hence there exists a symmetric
positive-definite matrix \(\mathsf P\) solving the Lyapunov equation
\begin{equation}
J_*^T\mathsf P+\mathsf P J_*=-\Id_4,
\label{eq:appD-Lyapunov-equation}
\end{equation}
where \(\Id_4\) is the \(4\times4\) identity.  Define
\begin{equation}
L_Q(\xi)=\xi^T\mathsf P\xi.
\label{eq:appD-quartic-Lyapunov}
\end{equation}
Using eqs.~\eqref{eq:appD-common-quartic-normal-form} and
\eqref{eq:appD-quartic-remainder},
\begin{equation}
L_Q'
\le
-\|\xi\|^2
+
2\|\mathsf P\|C_N\|\xi\|^3
+
2\|\mathsf P\|C_R\,\DT(G)\,\|\xi\|.
\label{eq:appD-Lyapunov-bound}
\end{equation}
Choose a sufficiently small radius \(r_*>0\) and the ellipsoid
\begin{equation}
\mathcal U_*
=
\{\xi:\;L_Q(\xi)<r_*^2\}
\label{eq:appD-Ustar}
\end{equation}
and then choose \(\delta_Q>0\) sufficiently small.  The right-hand side of
eq.~\eqref{eq:appD-Lyapunov-bound} is strictly negative on
\(L_Q=r_*^2\) whenever \(\DT(G)\le\delta_Q\).  Thus
\(\mathcal U_*\) is forward invariant as long as the cone mismatch remains
below \(\delta_Q\).  Theorem~\ref{thm:global-nonexpansion} supplies exactly
that bootstrap:
\begin{equation}
\DT(0)<\delta_Q
\quad\Longrightarrow\quad
\DT(\ell)\le\DT(0)<\delta_Q.
\label{eq:appD-quartic-bootstrap}
\end{equation}
Consequently the normalized quartic ratios remain in a compact
weak-coupling neighborhood for all infrared RG time, and
\begin{equation}
|\lambda_1|+|\lambda_2|+|\lambda_{12}|
=
O(S).
\label{eq:appD-quartic-O-S}
\end{equation}
This supplies the compact scalar-sector control entering the persistence
bounds below.

\subsection{Exact two-channel flow and invariant}
\label{appsub:exact-Sq}

At a common cone, the two-channel Pauli--Dirac flow is
\begin{subequations}
\label{eq:appD-ab-common}
\begin{align}
\dot a
&=
-\kappa a(7a^2-b^2),
\\
\dot b
&=
-\kappa b(7b^2-a^2).
\end{align}
\end{subequations}
Here \(\kappa>0\).
Set
\begin{equation}
x=a^2,\qquad
z=b^2,\qquad
S=x+z,\qquad
q=\frac{x-z}{x+z}.
\label{eq:appD-Sq-def}
\end{equation}
A direct calculation gives
\begin{subequations}
\begin{align}
\dot S
&=
-\kappa S^2(6+8q^2),
\label{eq:appD-S-common}\\
\dot q
&=
-8\kappa Sq(1-q^2).
\label{eq:appD-q-common}
\end{align}
\end{subequations}
For \(a,b>0\), one has \(|q|<1\).  The sign in
eq.~\eqref{eq:appD-q-common} points toward \(q=0\), and
\begin{equation}
6\kappa
\le
\frac{\dd}{\dd\ell}\frac1S
\le
14\kappa.
\label{eq:appD-recip-S-common}
\end{equation}
Thus \(S\asymp\ell^{-1}\).

The system also has an exact invariant.  Combining
eqs.~\eqref{eq:appD-S-common} and~\eqref{eq:appD-q-common},
\begin{equation}
\frac{\dd}{\dd\ell}
\left[
6\log|q|
-
7\log(1-q^2)
-
8\log S
\right]
=
0.
\label{eq:appD-invariant-log}
\end{equation}
Therefore
\begin{equation}
\frac{q^6}{(1-q^2)^7}
=
C S^8,
\qquad
C\ge0.
\label{eq:appD-invariant}
\end{equation}
For a nonzero initial imbalance this implies, as \(S\to0\),
\begin{equation}
q=O(S^{4/3})=O(\ell^{-4/3}).
\label{eq:appD-q-asymptotic}
\end{equation}
The balanced case \(q=0\) is invariant.  In either case
\begin{equation}
a^2\sim b^2\sim\frac{S}{2},
\end{equation}
and both edge activities have logarithmically divergent accumulated
integrals.

\subsection{Finite-mismatch persistence with explicit constants}
\label{appsub:persistence-constants}

We now turn the continuity argument into a uniform proposition.  By
Lemma~\ref{lem:appD-analytic-vertex}, the two-channel coupling flow in a
sufficiently small relative-cone neighborhood can be written in the
variables \(S,q\) as
\begin{subequations}
\begin{align}
\dot S
&=
-\kappa S^2
\left[
6+8q^2+\varepsilon_S(G,q)
\right],
\label{eq:appD-S-perturbed}\\
\dot q
&=
-\kappa S
\left[
8q(1-q^2)+\varepsilon_q(G,q)
\right],
\label{eq:appD-q-perturbed}
\end{align}
\end{subequations}
where, on
\begin{equation}
\DT(G)\le\delta,
\qquad
|q|\le q_*<1,
\label{eq:appD-control-domain}
\end{equation}
there are finite constants \(C_S,C_q\) such that
\begin{equation}
|\varepsilon_S(G,q)|\le C_S\delta,
\qquad
|\varepsilon_q(G,q)|\le C_q\delta.
\label{eq:appD-error-bounds}
\end{equation}
The constants are uniform because
Lemma~\ref{lem:appD-analytic-vertex} and the trapping region
\(\mathcal U_*\) of eq.~\eqref{eq:appD-Ustar} place the normalized quartic
couplings and kinetic data in a compact perturbative domain.

Choose any \(q_*\in(0,1)\) and let \(\delta_{\rm an}\) be a radius on which
the analytic representation above is valid.  It is sufficient to choose
\begin{equation}
0<\delta_0
<
\min\left\{
\delta_{\rm an},
\frac{4q_*(1-q_*^2)}{C_q},
\frac{3}{C_S}
\right\},
\label{eq:appD-delta-choice}
\end{equation}
with the corresponding constraint omitted if one of the displayed error
constants vanishes.  At \(q=q_*\),
\begin{equation}
8q_*(1-q_*^2)+\varepsilon_q
>
4q_*(1-q_*^2)>0,
\end{equation}
so \(\dot q<0\); at \(q=-q_*\) the sign is reversed and
\(\dot q>0\).  Hence the strip
\begin{equation}
|q|\le q_*
\label{eq:appD-q-strip}
\end{equation}
is forward invariant as long as
\(\DT\le\delta_0\).

The one-loop Pauli-channel equations have multiplicative form
\begin{equation}
\dot a=aF_a,\qquad
\dot b=bF_b,
\label{eq:appD-multiplicative}
\end{equation}
so \(a,b>0\) are preserved for finite RG time.  Positive couplings make all
cone-flow edge rates positive.  Theorem~\ref{thm:global-nonexpansion}
therefore gives
\begin{equation}
\DT(\ell)\le\DT(0).
\label{eq:appD-bootstrap-DT}
\end{equation}
An initial condition with
\(\DT(0)<\delta_0\) consequently remains inside the coefficient-control
neighborhood, closing the bootstrap independently of
Theorem~\ref{thm:finite-window}.

Within the invariant strip, define
\begin{equation}
c_1
=
\kappa(6-C_S\delta_0)>0,
\qquad
c_2
=
\kappa(6+8q_*^2+C_S\delta_0).
\label{eq:appD-c1c2}
\end{equation}
Then
\begin{equation}
-c_2S^2
\le
\dot S
\le
-c_1S^2,
\label{eq:appD-S-comparison}
\end{equation}
and integration yields
\begin{equation}
\frac{1}{
S(0)^{-1}+c_2\ell
}
\le
S(\ell)
\le
\frac{1}{
S(0)^{-1}+c_1\ell
}.
\label{eq:appD-S-bounds}
\end{equation}
Moreover,
\begin{equation}
a^2=\frac{1+q}{2}S,
\qquad
b^2=\frac{1-q}{2}S,
\end{equation}
so
\begin{equation}
\frac{1-q_*}{2}S
\le
a^2,b^2
\le
\frac{1+q_*}{2}S.
\label{eq:appD-edge-fractions}
\end{equation}
The lower bound and
eq.~\eqref{eq:appD-S-bounds} imply
\begin{equation}
\int^\infty a^2\,\dd\ell
=
\int^\infty b^2\,\dd\ell
=
\infty.
\label{eq:appD-infinite-activity}
\end{equation}

Finally write the two cone-flow coefficients as
\begin{equation}
\alpha_r
=
g_r^2\,K_r^{(B)}(G,Z),
\qquad
\beta_r
=
g_r^2\,K_r^{(F)}(G,Z),
\label{eq:appD-rate-prefactors}
\end{equation}
where \(g_1=a\), \(g_2=b\).  The loop derivations of
appendix~\ref{app:selfenergies} show that the kinetic prefactors are positive
and continuous.  Compactness therefore gives
\begin{equation}
0<K_-\le
K_r^{(B)},K_r^{(F)}
\le K_+<\infty.
\label{eq:appD-K-bounds}
\end{equation}
With
\begin{equation}
\chi=S,
\end{equation}
eqs.~\eqref{eq:appD-edge-fractions} and~\eqref{eq:appD-K-bounds} give
\begin{equation}
\frac{K_-(1-q_*)}{2}
\le
\frac{\alpha_r}{\chi},
\frac{\beta_r}{\chi}
\le
\frac{K_+(1+q_*)}{2}.
\label{eq:appD-normalized-rate-bounds}
\end{equation}
At the same time,
\begin{equation}
s(\ell)=\int_0^\ell S(u)\,\dd u\longrightarrow\infty.
\label{eq:appD-s-diverges}
\end{equation}
This proves the persistence statement used by
Corollary~\ref{cor:open-locking-basin}.

\subsection{Taylor-order proof for the topology theorem}
\label{appsub:topology-induction}

We now make the graph-distance counting in
section~\ref{sec:topology-onset} explicit.  Work at \(s=0\) and let \(x\)
be the common simple active direction of the two-plateau data.  Internal
plateau edges satisfy the equality condition of
Lemma~\ref{lem:appC-endpoint-equality}; hence the vector-valued action
\(G_vx\) agrees across such an edge at zeroth order.

For a lower-plateau node \(v\), define its arrival depth
\begin{equation}
\tau_v=d(v,U_x),
\end{equation}
and for an upper-plateau node define
\begin{equation}
\tau_v=d(v,L_x).
\label{eq:appD-tau}
\end{equation}
Thus \(\tau_v=1\) precisely when \(v\) is incident on an
opposite-plateau edge.

We use the following induction statement.  In the analytic active frame,
the directional displacement of a lower node is zero through order
\(\tau_v-1\) and has a positive order-\(\tau_v\) coefficient; for an upper
node the corresponding coefficient is negative.  When \(\tau_v=1\), the
claim follows from the strict endpoint sign on an edge joining opposite
plateaus.

Assume the statement holds for every node of arrival depth at most \(m\).
Let \(v\) have \(\tau_v=m+1\).  Every neighbor of \(v\) within its plateau
has zeroth-order equality with \(v\), and at least one neighbor \(w\) lies on
a shortest path with \(\tau_w=m\).  Differentiating the edge flux through
order \(m\), all terms containing derivatives of the normalized rate,
moving frame, or analytic active vector multiply lower-order mismatches and
vanish, exactly as in appendix~\ref{app:global-proofs}.  The first surviving
coefficient is therefore the coefficient arriving at \(w\) times the
positive edge multiplier
\begin{equation}
\mu_e=
\begin{cases}
\widehat\alpha_e(0), & \text{bosonic update},\\[2mm]
2\widehat\beta_e(0)c(R_e), & \text{fermionic update},
\end{cases}
\qquad
\mu_e>0.
\label{eq:appD-transfer-multiplier}
\end{equation}
Every shortest path has a product of positive multipliers, so contributions
from distinct shortest paths add with the same inward sign.  This completes
the induction.

Equivalently, if \(\mathcal P_v\) is the set of shortest paths from \(v\) to
the opposite plateau, the leading Taylor coefficient has the schematic
form
\begin{equation}
A_v
=
\sum_{\pi\in\mathcal P_v}
\sigma_{\partial\pi}
\prod_{e\in\pi_{\rm int}}\mu_e
>0,
\label{eq:appD-path-sum}
\end{equation}
where \(\sigma_{\partial\pi}>0\) is the strict source coefficient at the
first boundary edge and \(\pi_{\rm int}\) contains the saturated transfer
edges behind it.  Factorials depend on whether one writes derivatives or
Taylor coefficients and are positive, so they do not affect the sign or
order.

For an initially active lower--upper endpoint pair \((i,j)\), the branch
can move inward as soon as either endpoint is reached.  Hence
\begin{equation}
q_{ij}
=
\min\{\tau_i,\tau_j\}
=
\min\{d(i,U_x),d(j,L_x)\},
\label{eq:appD-qij}
\end{equation}
and
\begin{equation}
D_{ij}(0)-D_{ij}(s)
=
C_{ij}s^{q_{ij}}
+
O(s^{q_{ij}+1}),
\qquad
C_{ij}>0.
\label{eq:appD-pair-expansion}
\end{equation}

Branches that do not realize the initial Thompson diameter have a strictly
positive gap from the envelope and cannot become active at sufficiently
small \(s\).  Among initially active branches, those with smaller
\(q_{ij}\) move inward first.  Therefore the last branches remaining on the
maximum envelope are precisely those with
\begin{equation}
q=
\max_{(i,j)\in\mathcal A_x}q_{ij}.
\label{eq:appD-q-envelope}
\end{equation}
If several branches have this same order, the maximum selects the smallest
of their positive order-\(q\) coefficients.  Thus
\begin{equation}
\DT(0)-\DT(s)
=
C_qs^q+O(s^{q+1}),
\qquad
C_q>0,
\label{eq:appD-global-onset}
\end{equation}
which completes the proof of Theorem~\ref{thm:topology-onset}.

\subsection{\texorpdfstring{Explicit \(P_4\) coefficient}{Explicit P4 coefficient}}
\label{appsub:P4-coefficient}

For the four-node isotropic path
\begin{equation}
F_1-B_1-F_2-B_2
\end{equation}
let the initial scalar cone values be
\begin{equation}
(1,1,L,L),
\qquad
L>1,
\end{equation}
and let every bosonic and fermionic normalized edge rate be
\(\bar\alpha\) and \(\bar\beta\), respectively.  On the isotropic invariant
submanifold the fermion edge flow is
\begin{equation}
\dot a_F
=
2\bar\beta\,a_F\,h(b_B/a_F),
\label{eq:appD-isotropic-fermion-flow}
\end{equation}
where
\begin{equation}
h(r)
=
\frac{4(\sqrt r-1)}
{3\sqrt r(\sqrt r+1)^2},
\end{equation}
and the bosonic edge flow is
\begin{equation}
\dot b_B
=
\bar\alpha(a_F-b_B).
\label{eq:appD-isotropic-boson-flow}
\end{equation}

At \(s=0\), the far lower endpoint \(F_1\) and far upper endpoint \(B_2\)
have zero first derivative.  The central nodes satisfy
\begin{equation}
\dot B_1(0)
=
\bar\alpha(L-1),
\qquad
\dot F_2(0)
=
2\bar\beta L\,h(L^{-1}).
\label{eq:appD-central-first}
\end{equation}
Since \(D\mathcal H_{\Id}[K]=K/6\), the far lower endpoint has
\begin{equation}
\ddot F_1(0)
=
\frac{\bar\alpha\bar\beta}{3}(L-1),
\label{eq:appD-F1-second}
\end{equation}
whereas
\begin{equation}
\ddot B_2(0)
=
2\bar\alpha\bar\beta L\,h(L^{-1}).
\label{eq:appD-B2-second}
\end{equation}
The active far-end branch is
\begin{equation}
D(s)=\log\frac{B_2(s)}{F_1(s)}.
\end{equation}
Because both first derivatives vanish,
\begin{equation}
D''(0)
=
2\bar\alpha\bar\beta h(L^{-1})
-
\frac{\bar\alpha\bar\beta}{3}(L-1).
\label{eq:appD-Dsecond}
\end{equation}
Therefore
\begin{equation}
\DT(0)-\DT(s)
=
\bar\alpha\bar\beta
\left[
\frac{L-1}{6}-h(L^{-1})
\right]s^2
+
O(s^3).
\label{eq:appD-P4-final}
\end{equation}
Since \(L>1\) implies \(h(L^{-1})<0\), the coefficient is strictly positive.

\subsection{Numerical implementation for global locking}
\label{appsub:global-numerics}

The finite-matrix integrations are implemented in
\texttt{yukawa\_cone\_rg.py} and orchestrated by the reproduction driver
\texttt{reproduce\_figures.py}.
The spectral integral defining
\(\mathcal H(R)\) is evaluated with a fixed 64-point Gauss--Legendre rule on
\([0,1]\).  Matrix square roots and generalized eigenvalues are evaluated
with symmetric eigensolvers.  The full matrix ODE is integrated with
DOP853 using
\begin{equation}
{\tt rtol}=2\times10^{-9},
\qquad
{\tt atol}=2\times10^{-11}.
\label{eq:appD-DOP853-tols}
\end{equation}

For Figure~\ref{fig:global-locking}(a), the random seed is
\begin{equation}
20260804.
\end{equation}
Four \(3\times3\) SPD matrices are generated with logarithmic eigenvalue
radius \(1.4\) and reused unchanged for the alternating path \(P_4\) and
\(K_{2,2}\).  All normalized edge rates equal unity.  The interval
\(0\le s\le4\) is sampled at 301 points.  The common initial Thompson
diameter is
\begin{equation}
\DT(0)=2.1313012341140207,
\end{equation}
and the displayed final values are
\begin{equation}
\begin{aligned}
\mathcal D_T^{P_4}(4)&=0.827770121280161,\\
\mathcal D_T^{K_{2,2}}(4)&=0.19874555274700123.
\end{aligned}
\end{equation}

For panel~(b), the seed is
\begin{equation}
20260805.
\end{equation}
There are 16 displayed connected networks with two fermion and three boson
nodes.  A fixed connected backbone is supplemented by random edges, the
initial SPD logarithmic radius is \(1.8\), and each normalized edge rate is
drawn independently from
\begin{equation}
[e^{-0.35},e^{0.35}].
\end{equation}
The integrations use \(0\le s\le2\) with 241 samples.  All displayed trajectories have nonpositive sampled diameter increments; the
final median normalized diameter is
\begin{equation}
0.25113545058594416.
\end{equation}

\subsection{High-precision onset numerics}
\label{appsub:onset-numerics}

The onset calculation is performed on the invariant isotropic submanifold.
For a fermion--boson edge with scalar cone values \(u_F,u_B>0\) and unit
normalized rates,
\begin{equation}
\dot u_F
=
2u_F h(u_B/u_F),
\qquad
\dot u_B
=
u_F-u_B.
\label{eq:appD-scalar-path-ODE}
\end{equation}
Contributions from all incident edges are summed.

The paths \(P_4,P_6,P_8\) start with the symmetric two-level step
\begin{equation}
u_v(0)
=
\begin{cases}
1, & v<N/2,\\
8, & v\ge N/2.
\end{cases}
\label{eq:appD-onset-initial}
\end{equation}
The calculation uses \texttt{mpmath} at 60 decimal digits and a fixed
fourth-order Runge--Kutta step
\begin{equation}
\Delta s=10^{-5}
\end{equation}
through \(s=0.025\).  Power-law fits are performed over
\begin{equation}
3\times10^{-4}
\le
s
\le
3\times10^{-3}.
\label{eq:appD-fit-window}
\end{equation}
The fitted exponents are
\begin{center}
\begin{tabular}{c c c}
\hline
graph & predicted \(q\) & fitted power\\
\hline
\(P_4\) & 2 & 1.99238931847472\\
\(P_6\) & 3 & 2.9995006738958963\\
\(P_8\) & 4 & 3.994739533491443\\
\hline
\end{tabular}
\end{center}
The effective exponent plotted in
Figure~\ref{fig:topology-onset}(b) is computed from the numerical logarithmic
derivative
\begin{equation}
q_{\rm eff}(s)
=
\frac{\dd\log[\DT(0)-\DT(s)]}{\dd\log s}.
\end{equation}

The files \texttt{path\_onset\_high\_precision.csv} and
\texttt{path\_onset\_fit\_summary.csv} contain the underlying trajectories
and fit metadata.  The supplied scripts fix the seeds, tolerances, quadrature order, and figure
construction and regenerate the two quantitative figures directly.

\end{document}